\documentclass{article}
\usepackage[T1]{fontenc}
\usepackage[latin9]{inputenc}
\usepackage{amsthm,amsmath,amssymb}
\usepackage{booktabs}
\usepackage{multirow}
\usepackage{enumitem}
\usepackage[caption=false]{subfig}
\usepackage[english]{babel}
\usepackage{cleveref}
\usepackage{multirow}
\usepackage{pifont}
\usepackage{ulem}
\usepackage[margin=3cm]{geometry}
\usepackage{authblk}

\PassOptionsToPackage{normalem}{ulem}

\newtheorem{theorem}{Theorem}
\newtheorem{proposition}[theorem]{Proposition}
\newtheorem{lemma}[theorem]{Lemma}

\newtheorem{corollary}[theorem]{Corollary}
\newtheorem{example}[theorem]{Example}
\newtheorem{remark}[theorem]{Remark}
\newtheorem{assumption}[theorem]{Assumption}

\makeatletter
\usepackage{tikz}
\tikzset{%
	> = stealth, % arrow head style
	shorten > = 1pt, % don't touch arrow head to node
	auto,
	node distance = 3cm, % distance between nodes
	every edge/.append style = {thick}, % line style
}
\tikzstyle{node}=
[circle,draw=black,fill=black,minimum size=2mm, inner sep=0pt]

\definecolor{orange}{RGB}{255,127,0}
\newcommand{\boldall}[1]{\ifmmode\mathbf{#1}\else\textbf{\boldmath{#1}}\fi}

\title{Fractionally Subadditive Maximization
\\under an Incremental Knapsack Constraint\\ with Applications to Incremental Flows
\thanks{We acknowledge funding by the DFG through a short time project and subproject A007 of the CRC/TRR 154, as well as through grant DI 2041/2.}}

\author[1]{Yann Disser}
\author[2]{Max Klimm}
\author[1]{Annette Lutz}
\author[1]{David Weckbecker}

\affil[1]{TU Darmstadt, Germany, \tt \{disser|lutz|weckbecker\}@mathematik.tu-darmstadt.de}
\affil[2]{TU Berlin, Germany, \tt klimm@math.tu-berlin.de}

\date{}

\begin{document}
\maketitle

\normalem

\global\long\def\N{\mathbb{N}}%

\global\long\def\Q{\mathbb{Q}}%

\global\long\def\R{\mathbb{R}}%

\global\long\def\algscale{\textsc{Alg}_{\mathrm{scale}}}%

\global\long\def\optsol{S^{*}}%

\global\long\def\algsol{S^{\textrm{A}}}%

\global\long\def\algperm{\pi^{\textrm{A}}}%

\global\long\def\k{C}%

\global\long\def\gset{E}%

\global\long\def\base{b}%

\begin{abstract}
We consider the problem of maximizing a fractionally subadditive function
under an increasing knapsack constraint.
An incremental solution
to this problem is given by an order in which to include the elements
of the ground set, and the competitive ratio of an incremental solution
is defined by the worst ratio over all capacities relative to an optimum
solution of the corresponding capacity. We present an algorithm that
finds an incremental solution of competitive ratio at most $\max\{3.293\sqrt{M},2M\}$,
under the assumption that the values of singleton sets are in the
range $[1,M]$, and we give a lower bound of $\max\{2.618,M\}$ on
the attainable competitive ratio. In addition, we establish that our
framework captures potential-based flows between two vertices, and
we give a lower bound of~$\max\{2,M\}$ and an upper bound
of $2M$ for the incremental maximization of classical flows with
capacities in $[1,M]$ which is tight for the unit capacity case.
\end{abstract}

\section{Introduction}
\label{sec:introduction}

The decisions involved in large-scale infrastructure projects or in 
scheduling expensive investments usually have an impact over a prolonged period of time.
This paper examines the question how investment or construction decisions can be made over time as the total budget grows, in such a way that the resulting solution is good at every point in time.
This is a natural question for long-term infrastructure projects such as the construction of road networks, public transport systems, and energy networks. In addition, it is relevant for the investment decisions of businesses in manufacturing or distribution infrastructure. 

Mathematically, we model the above settings in terms of the \emph{incremental optimization problem}.
Formally, we are given a ground set $E$ of \emph{elements} that can be invested in. Each element $e \in E$ has a weight $w(e)$ that models the time or money that has to be spent on realizing the element. In the following, for a set $S \subseteq E$, we write $w(S) := \sum_{e \in S} w(e)$.
The value of having realized a subset $S \subseteq E$ of elements is given by a monotone objective function $f\colon 2^E \to \mathbb{R}_{\geq 0}$.
Given a capacity bound $C \in \mathbb{R}_{\geq 0}$, the maximum value that can be obtained with elements up to this total size is given by an optimum solution to the following mathematical optimization problem:
\begin{equation}
f^*(C) := \max \bigl\{f(S) \;\big\vert\; S\subseteq \gset, w(S)\leq C\bigr\}.\label{eq:generic-opt}\tag{P}
\end{equation}
Given $C \in \mathbb{R}_{\geq 0}$, we denote the optimal value of this optimization problem by $f^*(C)$ and a set $S \subseteq E$ for which the optimum is attained by $\optsol(C)$, where for the later we break ties in an arbitrary but fixed manner in order to obtain a unique set $\optsol(C)$.

We are interested in obtaining an \emph{incremental} solution to the optimization problem \eqref{eq:generic-opt} that yields a good value for \emph{all} capacity bounds $C \in \mathbb{R}_{\geq 0}$. Formally, an \emph{(incremental) solution} is an ordering $\pi = \bigl(e_{\pi(1)},e_{\pi(2)},\dots,e_{\pi(m)}\bigr)$ of the elements of the ground set $E$ with $m = |E|$.
For capacity $C \in \mathbb{R}_{\geq 0}$, let $\pi(C)$ be the items
contained in the maximal prefix of~$\pi$ that fits into the capacity~$C$,
i.e., $$\pi(C)=\left\{ e_{\pi(1)},e_{\pi(2)},\dots,e_{\pi(k)}\right\}$$
for some $k\in\N$ such that $\sum_{i=1}^{k}w(e_{\pi(i)})\leq C$
and either $k=m$ or $\sum_{i=1}^{k+1}w(e_{\pi(i)})>C$. We say that
the incremental solution $\pi$ is \emph{$\rho$-competitive} for some $\rho \geq 1$
if 
\[
f^{*}(C)\leq\rho\, f(\pi(C))\quad\text{for all }C \in \mathbb{R}_{\geq 0}.
\]
We call $\pi$ \emph{competitive} if it is $\rho$-competitive
for some constant $\rho\geq1$. 

As an example, let us consider the special case of the \emph{incremental maximum flow problem}. 
Here, the ground set $E$ corresponds to the set of edges of an undirected graph $G = (V,E)$ with two designated vertices $s,t \in V$. 
Each edge has a weight $w(e) \in \mathbb{R}_{\geq 0}$ and a capacity $u(e) \in \mathbb{R}_{\geq 0}$. 
The value $f(S)$ of a subset $S \subseteq E$ is defined as the maximum value of an $s$-$t$-flow in $G_S = (V,S)$.
Even in this special case, a competitive solution may fail to exist. For illustration, consider the graph in \Cref{subfig:inc-flow-1}. Every  solution~$\pi$ has to put edge~$a$ first in order to be competitive for $C=1$. On the other hand, every solution that puts edge~$a$ first is not better than $k$-competitive for $C=2$. Likewise, for the graph in \Cref{subfig:inc-flow-2}, every competitive solution has to put edge~$a$ first in order to be competitive for $C=1$, and every solution that puts edge~$a$ first is not better than $k$-competitive for $C=k$.

\begin{figure}
\centering
\subfloat[\label{subfig:inc-flow-1}]{
\begin{tikzpicture}[xscale=1.5,yscale=1]
\useasboundingbox (-.5,-1.7) rectangle (2.5,2);
\node[node,label=left:{$s$}] (s) at (0,1) {};
\node[node,label=right:{$t$}] (t) at (2,1) {};
\node[node] (v) at (1,-1) {};
\draw[thick,->] (s) to node[above] {$u(a)=1$} (t);
\draw[thick,->] (s) to node[below left] {$u(b)=k$}  (v);
\draw[thick,->] (v) to node [below right] {$u(c) = k$} (t);
\end{tikzpicture}
}
\hspace{1cm}
\subfloat[\label{subfig:inc-flow-2}]{
\begin{tikzpicture}
\useasboundingbox (-.5,-0.7) rectangle (4,2.2);
\node[node,label=left:{$s$}] (s) at (0,1) {};
\node[node,label=right:{$t$}] (t) at (3.5,1) {};
\draw[thick,->] (s) to[bend left] node[above] {$u(a) = w(a) = 1$} (t);
\draw[thick,->] (s) to[bend right] node[below] {$u(b) = w(b) = k$} (t);
\end{tikzpicture}
}
\caption{Two examples of the incremental flow problem where no $\rho$-competitive solution with $\rho < k$ exists.}	
\end{figure}
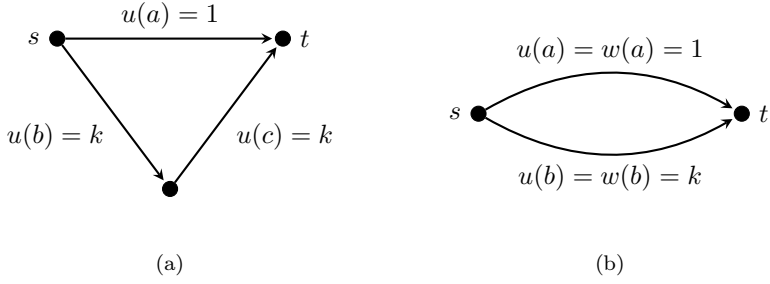

A closer inspection of these two examples reveals that there are (at least) two effects that prevent the existence of competitive solutions. The first is the \emph{complementarity} of elements. In the graph in \Cref{subfig:inc-flow-1}, edges~$b$ and $c$ are complementary in the sense that both edges together support a flow of $k$ while a single one of these edges alone cannot support any flow. For the graph in \Cref{subfig:inc-flow-2}, no two edges are complementary since the total flow supported by a subset of edges is here simply equal to the sum of the capacities of the edges. In this example, the non-existence of a competitive solution is caused by the fact that the edges are too heterogenous. More specifically, we have $f(\{a\}) = 1$, but $f(\{b\}) = k$, i.e., there are two singleton sets whose value differs by a factor of $k$.
%Consider a ground set of three elements $\gset=\{a,s,t\}$, a compact
%camera $a$ at a price of $w(a)=1$, a system camera $s$ at a price
%of $w(s)=2$, and a telephoto lens $t$ at a price of $w(t)=2$. For
%a subset $S\subseteq\gset,$ suppose that $f(S)=1$, if $S$ contains
%the compact camera, but not the system camera, $f(S)=2$ if it contains
%the system camera, but not the telephoto lens, $f(S)=3$ if it contains
%both the system camera and the telephoto lens, and $f(S)=0$, otherwise.
%If the buyer chooses to buy the compact camera first, they receive
%a utility of~$1$ which is optimal for the budget of~$1,$ i.e.,
%$f^{*}(1)=1$. However, if the buyer delays their purchase until they
%can afford the system camera, the value at time 1 is 0 but the value
%at time 2 is 2 which is optimal, i.e., $f^{*}(2)=2$. As this simple
%example illustrates, there is a tradeoff between the attained values
%at different times and, in particular, there is no purchase order
%that yields the optimal value at all times.
%For illustration, consider the example above where the value of the
%system camera is changed to $M\in\R_{\geq0}$. Since any competitive
%ordering is forced to buy the compact camera first in order to be
%competitive for $C=1$, no ordering can be better than $M$-competitive
%at time 2.

As we will show, these are essentially the only two effects that prevent the existence of competitive solutions. More specifically, we will make two assumptions that exclude the two effects shown in \Cref{subfig:inc-flow-1} and \Cref{subfig:inc-flow-2}. 
First, to avoid complementarities between elements, we assume that $f$ is \emph{fractionally subadditive}.
%Fractional subadditivity is a generalization of submodularity.
%and a standard assumption in the combinatorial auction literature
%(cf.~Nisan~\cite{Nisan/00}, Lehmann et al.~\cite{LehmannLehmannNisan/06}).
Formally, a function $f\colon2^{\gset}\to\R_{\geq0}$ is called \emph{fractionally
subadditive} if
\begin{align*}
f(A)\leq\sum_{i=1}^{k}\alpha_{i}f(B_{i}) &\quad\text{for all } A,B_{1},\dots,B_{k}\in2^{\gset} \text{ and all } \alpha_{1},\dots,\alpha_{k}\in\R_{\geq 0}\\ &\quad\quad\text{
such that } \sum_{i\in\left\{ 1,\dots,k\right\} :e\in B_{i}}\alpha_{i}\geq 1 \text{ for all } e\in A.
\end{align*}
Observe that fractional subadditivity implies regular
subadditivity, i.e., $f(A\cup B)\leq f(A)+f(B)$, but not vice-versa.
Second, to avoid that there exist singleton sets that differ too much in their values, we assume that there is a constant $M \in \R_{\geq 0}$, $M \geq 1$ such that  
 $f(\{e\})\in [1,M]$ for all $e\in\gset$.
We call such valuations \emph{$M$-bounded}.
%Observe that even for
%$1$-bounded valuations, complementarities between the items may prevent
%a competitive ordering. For illustration, adapt our example such that
%every non-empty subset of items yields a value of $1$ unless it contains
%both the system camera and the telephoto lens for which the value
%is $N>M$. Again, any competitive ordering has to put the compact
%camera first and then cannot be better than $N$-competitive for budget~$C=4$.
%To avoid this issue, we additionally require that~$f$ is fractionally
%subadditive. Fractional subadditivity is a generalization of submodularity
%and a standard assumption in the combinatorial auction literature
%(cf.~Nisan~\cite{Nisan/00}, Lehmann et al.~\cite{LehmannLehmannNisan/06}).
%Formally, a function $f\colon2^{\gset}\to\R_{\geq0}$ is called \emph{fractionally
%subadditive} if $f(A)\leq\sum_{i=1}^{k}\alpha_{i}f(B_{i})$ for all
%$A,B_{1},B_{2},\dots,B_{k}\in2^{\gset}$ and all $\alpha_{1},\alpha_{2},\dots,\alpha_{k}\in\R_{\geq0}$
%such that $\sum_{i\in\left\{ 1,\dots,k\right\} :e\in B_{i}}\alpha_{i}\geq1$
%for all $e\in A$. Observe that fractional subadditivity implies regular
%subadditivity, i.e., $f(A\cup B)\leq f(A)+f(B)$, but not vice-versa.
%In this paper, we address this tradeoff from a perspective of competitive
%analysis.

Summarizing the discussion, this paper considers incremental solutions to~(\ref{eq:generic-opt}) under the following assumptions.
\begin{assumption}
The function $f : 2^E \to \mathbb{R}_{\geq 0}$ has the following properties:
\begin{align}
&\text{$f$ is monotone, i.e., $f(A)\geq f(B)$ for $A\supseteq B$},\label{eq:MO}\tag{MO}\\
&\text{$f$ is $M$-bounded, i.e., $f(e)\in[1,M]$ for all $e\in E$},\label{eq:MB}\tag{MB}\\
&\text{$f$ is fractionally subadditive}.\label{eq:FS}\tag{FS}
\end{align}
\end{assumption}
Before stating our results, we illustrate the applicability of our
framework to different settings.
\begin{example}[Submodular objective]
A function $f : 2^E \to \mathbb{R}_{\geq 0}$ is called \emph{submodular} if, for all $A,B \in 2^E$, $f(A \cap B) + f(A \cup B) \leq f(A) + f(B)$.
% is fractionally subadditive.
It was shown by Lehmann et al. \cite{LehmannLehmannNisan/06} that every monotone submodular function is also fractionally subadditive.

As a consequence our framework captures, e.g., the \emph{maximum coverage
problem}, where we are given a weighted set of sets $E\subseteq2^{U}$
over a universe~$U$. Every element of $U$ has a value $v\colon U\to\mathbb{R}_{\geq0}$
associated with it, and $f(S)=v\left(\bigcup_{X\in S}X\right)$ for
all $S\subseteq E$ where we write $v(X):=\sum_{x\in X}v(x)$ for
a set $X\in2^{U}$. In this context, the $M$-boundedness condition
demands that $v(X)\in[1,M]$ for all $X\in E$. Further examples include
maximization versions of clustering and location problems.
\end{example}

\begin{example}[XOS objective]
\label{exa:XOS}An objective function $f\colon E\to\R$ is called
\emph{XOS} if it can be written as the pointwise maximum of modular
functions, i.e., there are $k\in\N$ and values $v_{e,i}\in\R$ for
all $e\in E$ and $i\in\left\{ 1,\dots,k\right\} $ such that
\begin{align*}
f(S)=\max \Biggl\{\sum_{e\in S}v_{e,i} \;\Bigg\vert\; i \in \{1,\dots,k\}\Biggr\}
\quad \text {for all $S\subseteq E$}.
\end{align*}
As shown by Feige~\cite{Feige/09}, the set of fractionally subadditive functions
and the set of XOS functions coincide. XOS
functions are a popular way to encode the valuations of buyers in
combinatorial auctions since they often give rise to a succinct representation
(cf.~Nisan~\cite{Nisan/00} and Lehman et al.~\cite{LehmannLehmannNisan/06}).
\end{example}

\begin{example}[Weighted rank function of an independence system]
An independence system is a tuple~$(E,\mathcal{I})$, where $\emptyset\in\mathcal{I}$ and $\mathcal{I\subseteq}2^{E}$
is closed under taking subsets, i.e., $A \in \mathcal{I}$ whenever $A \subseteq B$ and $B \in \mathcal{I}$. For
a given weight function $v \colon E\to\mathbb{R}_{\geq0}$, the weighted
rank function of $(E,\mathcal{I})$ is given by
\begin{align*}
f(S)=\max \Biggl\{\sum_{e\in I}v(e) \; \Bigg\vert\; I \in \mathcal{I}\cap2^{S}\Biggr\} \quad \text{ for all $S\subseteq E$}.
\end{align*}
As shown by Amanatidis et al.~\cite{AmanatidisBirmpasMarkakis/16}, the weighted rank function of an independence system is fractionally
subadditive.

This setting captures well-known problems such as \emph{weighted $d$-dimensional
matching} for any $d\in\N$. Here, we are given sets $V_1,\dots,V_d$ such that $E \subseteq V_1 \times \cdots \times V_d$ and a function $v : E \to \mathbb{R}_{\geq 0}$. The value $f(S)$ is the defined as the maximum weight of a $d$-dimensional matching in $S$, i.e,
\begin{align*}
f(S) = \max \Biggl\{ \sum_{e \in M} v(e) \;\Bigg\vert\; &M \subseteq S \text{ such that } v_i \neq v_i' \text{ for all $i \in \{1,\dots,d\}$} \\[-8pt]
&\quad\text{ for all } e \!=\! (v_1,\dots,v_d), e' \!=\! (v_1',\dots,v_d') \in M \text{ with } e \neq e' \Biggr\}	.
\end{align*}

In a similar vein, this setting also includes weighted set packing and weighted maximum independent set.
\end{example}

\begin{example}[Potential-based flows]
\label{exa:supply}
Consider a variant of the incremental flow problem on parallel edges as in \Cref{subfig:inc-flow-2}. 
As before, every edge~$e$ has an capacity $u(e) \in \mathbb{R}_{\geq 0}$. In addition, we are given a continuous and strictly increasing potential-loss function $\psi : \mathbb{R} \to \mathbb{R}$ with $\lim_{x \to \infty} 
\psi(x) = \infty$ that describes the physical properties of the network. Every edge~$e$ has a \emph{resistance} $\beta(e) \in \mathbb{R}_{\geq 0}$.
A vector $x \in \mathbb{R}_{\geq 0}^E$ is a flow if $x_e \leq u(e)$ for all $e \in E$, and it is called a potential-based flow if there are vertex potentials $p_{s}, p_t \in \mathbb{R}_{\geq 0}$ such that
\begin{align*}
p_{s}-p_{t}=\beta(e)\psi(f(e)) \quad\text{ for all $e\in E$}.
\end{align*}
The potentials correspond to physical properties at the nodes such as pressures or voltages; different choices of $\psi$ allow to
model gas flows, water flows, and electrical flows, see Gro\ss\ et al.~\cite{GrossPSSS/19}.
In our incremental framework, $w\colon\gset\to\mathbb{R}_{\geq0}$
are interpreted as construction costs of pipes or cables and the objective
is to maximize the flow from~$s$ to~$t$ in terms of the objective
\begin{equation*}
f(S)=\max \Biggl\{\sum_{e\in T} \psi^{-1}\biggl(\frac{p}{\beta(e)}\biggr)\,\Bigg\vert\, {T \subseteq S}, p \in \mathbb{R}_{\geq 0} \:\text{with}\: \psi^{-1}\biggl(\frac{p}{\beta(e)}\biggr)\!\!\leq\! u(e)\:\text{for all}\,e\in T\Biggl\},\label{eq:pb_flow}
\end{equation*}
where $p:=p_{s}-p_{t}$ and $u(e)$. The value
of the objective is the maximum value of a feasible potential-based
$s$-$t$-flow where we allow turning off edges in $S\setminus T$ in order to make $f$ monotone. The $M$-boundedness
condition corresponds to the assumption that $u(e)\in[1,M]$.
As we will show in \Cref{prop:pot-based_max-flow_is_fractionally-subadditive}, this objective is fractionally
subadditive.
 
\end{example}

\subsection{Our results}

Our main results are bounds on the best possible competitive ratio
for incremental solutions to~(\ref{eq:generic-opt}) for objectives
satisfying \eqref{eq:MO}, \eqref{eq:MB}, and \eqref{eq:FS}. In other words, we bound the loss
in solution quality that we have to accept when asking for incremental
solutions that optimize for all capacities simultaneously. Note that,
as customary in online optimization, we do not impose restrictions
on the computational complexity of finding incremental solutions.
To state our result, we denote by $\varphi=\frac{1}{2}(1+\sqrt{5}) \approx 1.618$ the golden ratio.
\begin{theorem}
\label{thm:upper_and_lower_bound_for_M-bounded}For monotone, $M$-bounded,
and fractionally subadditive objectives, the best-possible~$\rho$\linebreak
for which the optimization problem~\eqref{eq:generic-opt} admits a $\rho$-compet\-itive
solution satisfies\linebreak $\rho\in\bigl[\max\bigl\{\varphi+1,M\bigr\},\max\bigl\{3.293\sqrt{M},2M\bigr\}\bigr].$
\end{theorem}

In particular, for $M\geq2.71$, the best possible competitive ratio
is between~$M$ and~$2M$, while the bounds for $1$-bounded objectives
simplify as follows.
\begin{corollary}
\label{thm:upper_and_lower_bound_for_1-bounded}For monotone, $1$-bounded,
and fractionally subadditive objectives, the best~$\rho$ for which
the optimization problem~\eqref{eq:generic-opt} admits a $\rho$-competitive solution satisfies $\rho\in \bigl[\varphi+1,3.293\bigr]$.
\end{corollary}

Our upper bounds are shown by an algorithm that uses a simultaneous capacity-
and value-scaling approach. In each phase, we increase our capacity
and value thresholds and pick the smallest capacity for which the
optimum solution exceeds our thresholds. This solution is then assembled
by adding one element at a time in a specific order. The order is
chosen based on a primal-dual LP formulation that relies on fractional
subadditivity.

For the definition of the algorithm, we need access to two oracles. On the
one hand, we need oracle access to the optimal solution of a given
capacity; on the other hand, we need access to an XOS oracle. More
information on this can be found in Remark~\ref{rem:oracle_access}.

In Section~\ref{sec:ub}, we describe our algorithmic approach in
detail and give a proof of the upper bound. In Section~\ref{sec:lb},
we complement our result with two lower bounds. As an additional motivation,
in Section~\ref{sec:flows}, we show that our framework captures
potential-based flows as described in Example~\ref{exa:supply}.
In this context, a $1$-bounded objective corresponds to unit capacities.
As a contrast, we also show that classical $s$-$t$-flows with capacities
in $[1,M]$ admit $2M$-competitive incremental solutions, and this
is best-possible for the unit capacity case.

\subsection{Related Work}

Bernstein et al.~\cite{BernsteinDisserGrossHimburg/20} considered
a closely related framework for incremental maximization. Their framework
assumes a growing cardinality constraint, which is a special case
of our problem in~(\ref{eq:generic-opt}) when all elements $e\in E$
have unit weight $w(e)=1$. A natural incremental approach for a growing
cardinality constraint is the greedy algorithm that includes in each
step the element that increases the objective the most. This algorithm
is well known to yield a $e/(e-1)$ approximation for submodular objectives~\cite{NemhauserWolseyFisher/78}.
Several generalizations of this result to broader classes of functions
are known. Recently, Disser and Weckbecker~\cite{DisserWeckbecker/20}
unified these results by giving a tight bound for the approximation
ratio of the greedy algorithm for $\gamma$-$\alpha$-augmentable
functions\footnote{A function is called $\gamma$-$\alpha$-augmentable if, for all sets
$X,Y\subseteq\gset$, there exists $y\in Y$ with $f(X\cup\{y\})-f(X)\geq(\gamma f(X\cup Y)-\alpha f(X))/|Y|.$}, which interpolates between known results for weighted rank functions
of independence systems of bounded rank quotient, functions of bounded
submodularity ratio, and $\alpha$-augmentable functions. Sviridenko~\cite{Sviridenko/04}
showed that for a submodular function under a knapsack constraint,
the greedy algorithm yields a $(1-1/e)$-approximation when combined
with a partial enumeration procedure. This approximation guarantee
is best possible as shown by Feige~\cite{Feige98}. Yoshida~\cite{Yoshida19}
generalized the result of Sviridenko to submodular functions of bounded
curvature.

Another closely related setting is the robust maximization of a modular
function under a knapsack constraint. Here, the capacity of the knapsack
is revealed in an online fashion while packing, and we ask for a packing
order that guarantees a good solution for every capacity. Megow and
Mestre~\cite{MegowM13} considered this setting under the assumption
that we have to stop packing once an item exceeds the knapsack capacity
and presented a polynomial time algorithm that has an instance-sensitive
near-optimal competitive ratio. Navarra and Pinotti~\cite{NavarraP17}
added the mild assumption that all items fit in the knapsack and devised
competitive solutions for this model. Disser et al.~\cite{DisserKMS17}
allowed to discard items that do not fit and showed tight competitive
ratios for this case. Kawase et al.~\cite{KawaseSF19} studied a
generalization of this model in which the objective is submodular
and devised a randomized competitive algorithm for this case.
Klimm and Knaack~\cite{KlimmK22} gave a deterministic competitive algorithm with improved competitive ratio for this case.
Since the models in \cite{DisserKMS17,KawaseSF19,KlimmK22} allow to discard items, these competitive ratios do not
translate to our model. Kobayashi and Takazawa~\cite{KobayashiT17}
studied randomized strategies for cardinality robustness in the knapsack
problem. Other online versions of the knapsack problem assume that
items are revealed over time, e.g., see Matchetti-Spaccamela and Vercellis~\cite{Marchetti-SpaccamelaV95}.
Thielen et al.~\cite{ThielenTW16} combined both settings and assumed
that items appear over time while the capacity grows. An overview
over a selection of the aforementioned results can be found in \Cref{fig:CR_in_literature}.

\begin{figure}
\begin{center}
\begin{tabular}{l l c c c c}
\toprule
	\textbf{objective} \boldmath$f$\hspace{1cm} & \textbf{setting} & \multicolumn{2}{c}{\textbf{lower bound}} & \multicolumn{2}{c}{\textbf{upper bound}} \\
	\midrule
	additive &  & $\infty$ & \cite{MegowM13} &  \\
	additive & with discarding & $2$ & \cite{DisserKMS17} & $2$ & \cite{DisserKMS17} \\
	additive & with discarding, $f=w$ & $\varphi$ & \cite{DisserKMS17} & $\varphi$ & \cite{DisserKMS17} \\
	additive & $C\geq\max w(e)$ & $2$ & \cite{NavarraP17} & $2$ & \cite{NavarraP17} \\
	additive & $C\geq\max w(e)$, $f=w$ & $\varphi$ & \cite{NavarraP17} & $1.756$ & \cite{NavarraP17} \\[6pt]
	submodular & with discarding & 2 & \cite{DisserKMS17} & 2.794 & \cite{KlimmK22}\\[6pt]
	frac.~subadditive & $f(\{e\})=1$ & $\varphi+1$ &  & $3.293$ & \\[6pt]
	accountable & $w(e)=1$ & $ 2.183$ & \cite{BernsteinDisserGrossHimburg/20} & $\varphi+1$ & \cite{BernsteinDisserGrossHimburg/20} \\
	\bottomrule
\end{tabular}
\end{center}

\caption{\label{fig:CR_in_literature}
Competitive ratios for deterministic incremental
maximization of a monotone objective under a knapsack constraint in
different settings where $\varphi=\frac{1}{2}(1+\sqrt{5}) \approx 1.617$ is the golden ratio. Note that additivity implies fractional subadditivity
(XOS), which implies accountability.}
\end{figure}

In terms of incremental minimization, Lin et al.~\cite{LinNRW10}
introduced a general framework, based on a problem-specific augmentation
routine, that subsumes several earlier results. A maximization problem
with growing cardinality constraint that received particular attention
is the so-called robust matching problem introduced by Hassin and
Rubinstein~\cite{HassinR02}. Here, we ask for a weighted matching
such that the heaviest $k$ edges of the matching approximate a maximum
weight matching of cardinality~$k$, for all cardinalities~$k$.
Hassin and Rubinstein~\cite{HassinR02} gave tight bounds on the
deterministic competitive ratio of this problem, and Matuschke et
al.~\cite{MatuschkeSS18} gave bounds on the randomized competitive
ratio. Fujita et al.~\cite{FujitaKM13} and Kakimura et al.~\cite{KakimuraM13}
considered extensions of this problem to independence systems. A similar
variant of the knapsack problem where the $k$ most valuable items
are compared to an optimum solution of cardinality $k$ was studied
by Kakimura et al.~\cite{Kakimura12}.

Incremental optimization has also been considered from an offline
perspective, i.e., without uncertainty in items or capacities. Kalinowski
et al.~\cite{KalinowskiMatsypuraSavelsbergh/15} and Hartline and
Sharp~\cite{HartlineS07} considered incremental flow problems where
the average flow over time needs to be maximized (in contrast to the
worst flow over time). Anari et al.~\cite{AnariHNPST19} and Orlin
et al.~\cite{OrlinSU18} considered general robust submodular maximization
problems.

The class of fractionally subadditive valuations was introduced by Nisan
\cite{Nisan/00} and Lehman et al.~\cite{LehmannLehmannNisan/06}
under the name of XOS-valuations as a compact way to represent the
utilities of bidders in combinatorial auctions. In a combinatorial
auction, a set of elements $\gset$ is auctioned off to a set of $n$
bidders who each have a private utility function $f_{i}\colon2^{\gset}\to\R_{\geq0}$.
In this context, a natural question is to maximize social welfare,
i.e., to partition $\gset$ into sets $E_{1},E_{2},\dots,E_{n}$ with
the objective to maximize $\sum_{i=1}^{n}f_{i}(E_{i})$. Dobzinski
and Schapira~\cite{DobzinskiSchapiry/06} gave a $(1-1/e)$-approximation
for this problem.

\section{Upper bound\label{sec:ub}}

In the following, we fix a ground set $E$, a monotone,  $M$-bounded
and fractionally subadditive objective $f\colon2^{\gset}\to\mathbb{R}_{\geq0}$,
and weights $w\colon\gset\to\mathbb{R}_{\geq0}$. We present a refined
variant of the incremental algorithm introduced in~\cite{BernsteinDisserGrossHimburg/20}.
On a high level, the idea is to consider optimum solutions of increasing
sizes, and to add all elements in these optimum solutions one solution
at a time. By carefully choosing the order in which we add elements
of a single solution, we ensure that elements contributing the most
to the objective are added first. In this way, we can guarantee that
either the solution we have assembled most recently, or the solution
we are currently assembling provides sufficient value to stay competitive.
While the algorithm of~\cite{BernsteinDisserGrossHimburg/20} only
scales the capacity, our algorithm $\algscale$ simultaneously scales
capacities and solution values. In addition, we use a more sophisticated
order in which we assemble solutions, based on a primal-dual LP formulation.
We now describe our approach in detail.

Let $\lambda\approx3.2924$ be the unique real root of the equation
\[
\mbox{\ensuremath{0=\lambda^{7}-2\lambda^{6}-3\lambda^{5}-3\lambda^{4}-3\lambda^{3}-2\lambda^{2}-\lambda-1}},
\]
$\smash{\delta:=\frac{\lambda^{3}}{\lambda^{2}+1}}\approx3.0143$
and
\begin{align}
\label{eq:definition-rho}
\rho:=\max\bigl\{\lambda\sqrt{M},2M \bigr\}.
\end{align}
Algorithm $\algscale$
operates in phases of increasing capacities 
$C_{1},\dots,C_{N}\in\mathbb{R}_{\geq0}$
with
\begin{align*}
\k_{1} & := \min_{e\in E}w(e), \\ 
\k_{i} & :=\min\bigl\{\k\geq\delta\k_{i-1} \;\big\vert\; f^{*}(\k)\geq\rho f^{*}(\k_{i-1}) \bigr\} \quad \text{ for all $i \in \N$},
\end{align*}
where we set $\min\emptyset=\sum_{e\in E}w(e)$. Let
$N\in\mathbb{N}$ be the minimal index such that $C_{N}=\sum_{e\in E}w(e)$.
In phase $i\in\{1,\dots,N\}$, $\algscale$ adds the elements of the
set $\optsol(C_{i})$ one at a time. Recall that $\optsol(C)$ is
the optimum solution to~(\ref{eq:generic-opt}) for capacity~$C$.
We may assume that previously added elements are added again (without
any benefit), since this only hurts the algorithm.

To specify the order in which the elements of $\optsol(\k_{i})$
are added, consider the following linear program ($\mathrm{LP}_{X}$)
parameterized by $X\subseteq\gset$ (cf.~\cite{Feige/09}):
\begin{alignat*}{3}
\min &  & \sum_{B\subseteq\gset}\alpha_{B}f(B)\\
\mathrm{s.t.} &  & \sum_{B\subseteq\gset:e\in B}\alpha_{B} & \geq1, & \quad & \text{ for all } e\in X,\\
 &  & \alpha_{B} & \geq0, &  & \text{ for all } B\subseteq\gset,
\end{alignat*}
and its dual
\begin{alignat*}{3}
\max &  & \sum_{e\in X}\gamma_{e}\\
\mathrm{s.t.} &  & \sum_{e\in B}\gamma_{e} & \leq f(B), & \quad & \text{ for all } B\subseteq\gset,\\
 &  & \gamma_{e} & \geq0, &  & \text{ for all } e\in X.
\end{alignat*}
Fractional subadditivity of $f$ translates to $f(X)\leq\sum_{B\subseteq\gset}\alpha_{B}f(B)$
for all $\alpha\in\R^{2^{\gset}}$ feasible for ($\mathrm{LP_{X}}$).
The solution $\alpha^{*}\in\smash{\R^{2^{\gset}}}$ with $\alpha_{X}^{*}=1$
and $\alpha_{B}^{*}=0$ for $X\neq B\subseteq\gset$ is feasible and
satisfies\linebreak $f(X)=\sum_{B\subseteq\gset}\alpha_{B}^{*}f(B)$. Together
this implies that~$\alpha^{*}$ is an optimum solution to ($\mathrm{LP_{X}}$).
By strong duality, there exists an optimum dual solution $\gamma^{*}(X)\in\R^{\gset}$
with 
\begin{equation}
f(X)=\sum_{e\in X}\gamma_{e}^{*}(X).\label{eq:optimum_dual_solution}
\end{equation}

In phase 1, the algorithm $\algscale$ adds the unique element in
$\optsol(\k_{1})$. In phase 2, $\algscale$ adds an element $e \in\optsol(\k_{2})$
first that maximizes $\gamma_{e}^{*}$ and the other elements in an
arbitrary order. In phase $i\in\{3,4,\dots,N\}$, $\algscale$ adds
the elements of $\optsol(\k_{i})$ in an order $(e_{1},\dots,e_{|\optsol(\k_{i})|})$
such that, for all $j\in\{1,\dots,|\optsol(\k_{i})|-1\}$,
\begin{equation}
\frac{\gamma_{e_{j}}^{*}(\optsol(\k_{i}))}{w(e_{j})}\geq\frac{\gamma_{e_{j+1}}^{*}(\optsol(\k_{i}))}{w(e_{j+1})}.\label{eq:density_decreasing}
\end{equation}
For $C\in[0,\k_{i}]$ and with $j:=\max\{j\in\{1,\dots,|\optsol(\k_{i})|\}\mid w(\{e_{1},\dots,e_{j}\})\leq C\}$,
we denote the prefix
of $\optsol(C_{i})$ of capacity $C$ by $\optsol(\k_{i},C):=\{e_{1},\dots,e_{j}\}$. Furthermore, by~$\algperm$,
we refer to the permutation of $\gset$ that represents the order
in which the algorithm $\algscale$ adds the elements of $\gset$.

Roughly, we show that this algorithm is competitive as follows: In
the first phase $\algscale$ obviously performs optimally. In all
other phases, the solution added in the previous phase is large enough
to be competitive until the solution added currently has a larger
value. From this point until the end of the phase, the current solution
is competitive.
\begin{remark}
\label{rem:oracle_access}In the construction of our algorithm, we
assume to have oracle access to an optimum solution $\optsol(C)$
of a given capacity $C\in\R_{\geq0}$. Finding such a solution may
not be possible in polynomial time. Badanidiyuru et al.~\cite{BadanidiyuruDO12},
give a $(2+\varepsilon)$-approximation algorithm
and show that no polynomial time algorithm can have an approximation
ratio of less than $2$, unless $\mathsf{P}=\mathsf{NP}$. Our algorithm
$\algscale$ can use an $\alpha$-approximation oracle instead of
an oracle for the optimum solution, for a loss of factor $\alpha$
in its competitive ratio. Furthermore, we assume to have access to
an XOS oracle. For a given set $X\subseteq\gset$ and $x\in X$, an
XOS oracle gives the value of $x$ within the set $X$, which corresponds
to the solution of the dual LP mentioned above. Instead of an XOS
oracle, our algorithm can use an $\beta$-approximation oracle for
a loss of factor $\beta$ in its competitive ratio.
\end{remark}

We first show that the dual variables $\gamma_{e}^{*}(X)$ associate
a contribution to the overall objective to each element $x\in X$,
and that this association is consistent under taking subsets of $X$.
\begin{lemma}
\label{lem:subset-value_and_dual-variables}Let $X\subseteq Y\subseteq\gset$.
Then,
\begin{align}
\label{eq:subset-value_and_dual-variables}
f(X)\geq\sum_{e\in X}\gamma_{e}^{*}(Y).
\end{align}
\end{lemma}

\begin{proof}
Since $\gamma^{*}(Y)$ is a feasible solution for the dual of ($\mathrm{LP}_{Y}$),
it is also a feasible solution for the dual of ($\mathrm{LP}_{X}$).
Thus, since $\gamma^{*}(X)$ is an optimum solution of ($\mathrm{LP}_{X}$),
\[
\sum_{e\in X}\gamma_{e}^{*}(Y)\leq\sum_{e\in X}\gamma_{e}^{*}(X)\overset{\eqref{eq:optimum_dual_solution}}{=}f(X).
\]
\end{proof}

The following lemma establishes that the order in which we add the
elements of each optimum solution are decreasing in density, in an
approximate sense.
\begin{lemma}
\label{lem:estimate_value-of-subset}Let $C,C'\in\R_{\geq0}$ with
$C\leq C'\leq w(\gset)$. Then
\begin{align}
\label{eq:estimate_value-of-subset}
f^{*}(C')\leq\frac{C'}{C}\bigl(f(\optsol(C',C))+M\bigr).
\end{align}
\end{lemma}

\begin{proof}
If $\optsol(C')=\optsol(C',C)$, the statement holds trivially.
Suppose that $|\optsol(C')|>|\optsol(C',C)|$. Let $j:=|\optsol(C',C)|$,
and let $\optsol(C')= \bigl\{e_{1},\dots,e_{|\optsol(C')|}\bigr\}$ such that
(\ref{eq:density_decreasing}) holds. Note that, by definition, $\optsol(C',C)=\{e_{1},\dots,e_{j}\}$
and
\begin{equation}
w(\{e_{1},\dots,e_{j}\})\leq C<w(\{e_{1},\dots,e_{j+1}\}).\label{eq:set-weights_and_C}
\end{equation}
We have
\begin{eqnarray*}
f^{*}(C') & \overset{\eqref{eq:optimum_dual_solution}}{=} & \sum_{i=1}^{|\optsol(C')|}w(e_{i})\frac{\gamma_{e_{i}}^{*}(\optsol(C'))}{w(e_{i})}\\
 & \overset{\eqref{eq:density_decreasing}}{\leq} & \Biggl(\sum_{i=1}^{j+1}\gamma_{e_{i}}^{*}(\optsol(C'))\Biggr)+\frac{\sum_{i=1}^{j+1}w(e_{i})}{w(\{e_{1},\dots,e_{j+1}\})}\sum_{i=j+2}^{|\optsol(C')|}w(e_{i})\frac{\gamma_{e_{j+1}}^{*}(\optsol(C'))}{w(e_{j+1})}\\
 & \overset{\eqref{eq:density_decreasing}}{\leq} & \Biggl(\sum_{i=1}^{j+1}\gamma_{e_{i}}^{*}(\optsol(C'))\Biggr)+\frac{\Bigl(\sum_{i=1}^{j+1}\gamma_{e_{i}}^{*}(\optsol_{C'})\Bigr)}{w(\{e_{1},\dots,e_{j+1}\})}\sum_{i=j+2}^{|\optsol(C')|}w(e_{i})\\
 & \overset{\eqref{eq:set-weights_and_C}}{<} & \Biggl(\sum_{i=1}^{j+1}\gamma_{e_{i}}^{*}(\optsol(C'))\Biggr)+\frac{\Bigl(\sum_{i=1}^{j+1}\gamma_{e_{i}}^{*}(\optsol(C'))\Bigr)}{C}(C'-C)\\
 & = & \frac{C'}{C}\Biggl[\Biggl(\sum_{i=1}^{j}\gamma_{e_{i}}^{*}(\optsol(C'))\Biggr)+\gamma_{e_{j+1}}^{*}(\optsol(C'))\Biggr]\\
 & \overset{\eqref{eq:subset-value_and_dual-variables}}{\leq} & \frac{C'}{C}\bigl(f(\{e_{1},\dots,e_{j}\})+f(\{e_{j+1}\})\bigr)\\
 & \leq & \frac{C'}{C}\bigl(f(\optsol(C',C))+M \bigr),
\end{eqnarray*}
completing the proof.
\end{proof}

Since every set $S\subseteq\gset$ with $w(S)\leq C$ satisfies $f(S)\leq f^{*}(C)$,
and since we have $w(\optsol(C',C))\leq C$, we immediately obtain
the following.
\begin{corollary}
\label{cor:estimation_value-of-two-optimum-solutions}Let $C,C'\in\R_{\geq0}$
with $C\leq C'\leq w(\gset)$. Then
\begin{align}
\label{eq:estimation_value-of-two-optimum-solutions}
f^{*}(C')\leq \frac{C'}{C} \bigl(f^{*}(C)+M\bigr).
\end{align}
\end{corollary}

With this, we are now ready to show the upper bound of our main result.
\begin{theorem}
For $\rho=\max \bigl\{\lambda\sqrt{M},2M\bigr\}\approx\max \bigl\{3.2924\sqrt{M},2M\bigr\}$,
the incremental solution computed by $\algscale$ is $\rho$-competitive.\label{thm:upper_bound_M}
\end{theorem}

\begin{proof}
We have to show that, for all sizes $C\in\R_{\geq0}$, we have $f^{*}(C)\leq\rho f(\algperm(C))$.
We will do this by analyzing the different phases of the algorithm.
Observe that, for all $i\in\{2,\dots,N-1\}$, we have
\begin{align}
f^{*}(\k_{i}) & \geq  \rho f^{*}(\k_{i-1}) 
  \geq  \rho^{i-1}f^{*}(\k_{1})  \overset{\textrm{(MB)}}{\geq}  \rho^{i-1} \geq  (\lambda\sqrt{M})^{i-1},\label{eq:value_in _phase_i}
\end{align}
where for the first inequality, we use the definition of the algorithm $\algscale$, and for the last inequality we use the definition of $\rho$ in \eqref{eq:definition-rho}.

In phase~1, we have $C\in(0,\k_{1}]$. Since $\k_{1}$ is the minimum
weight of all elements and we start by adding~$\optsol(\k_{1})$,
i.e., the optimum solution of size~$C_{1}$, the value $\algperm(C)$
is optimal.

Consider phase~2, and suppose $C\in(\k_{1},\k_{2})$. If $\k_{2}>\delta\k_{1}$
holds, then $\k_{2}$ is the smallest value such that $f^{*}(\k_{2})\geq\rho f^{*}(\k_{1})$,
i.e., by monotonicity of $f$, we have
\begin{align*}
f(\algperm(C))\geq f(\algperm(\k_{1}))=f^{*}(\k_{1})>\frac{1}{\rho}f^{*}(C).
\end{align*}
Now assume $\k_{2}\leq\delta\k_{1}$. If $C\in(\k_{1},3\k_{1})$,
i.e., any solution of size $C$ cannot contain more than two elements,
or if $C\in(\k_{1},\k_{2})$ and~$\optsol(\k_{2})$ contains at
most~2 elements, by fractional subadditivity and $M$-boundedness
of $f$, we have $f^{*}(C)\leq|\optsol_{\k_{2}}|M\leq2M$ and thus,
\begin{align*}
f(\algperm(C))\geq f^{*}(\k_{1})\geq1\geq\frac{1}{2M}f^{*}(C)\geq\frac{1}{\rho}f^{*}(C).
\end{align*}
Now suppose $C\in[3\k_{1},\k_{2})$ and that $\optsol(\k_{2})$ contains
at least~3 elements. The prefix $\algperm(\k_{1}+\k_{2})$ contains
all elements from $\optsol(\k_{1}) \cup \optsol(\k_{2})$, the prefix
$\algperm(\k_{2})=\algperm(\k_{1}+\k_{2}-\k_{1})$ contains at least
all but one elements of $\optsol(\k_{2})$, and the prefix $\algperm(\k_{2}-\k_{1})$
contains at least all but~2 elements of $\optsol(\k_{2})$ because
the weight of any element is at least $\k_{1}$. Since 
\begin{align*}
C\geq3\k_{1}>(\delta-1)\k_{1}\geq\k_{2}-\k_{1},
\end{align*}
$\algperm(C)$ contains at least all but~2 elements from $\optsol(\k_{2})$.
Recall that in phase~2, the algorithm adds the element $e\in\optsol(\k_{2})$
that maximizes $\gamma_{e}^{*}$ first. Therefore, and because $|\optsol(\k_{2})|\geq3$,
we have $f(\algperm(C))\geq\frac{1}{3}f(\optsol(\k_{2}))\geq\frac{1}{\rho}f^{*}(C)$.

Consider phase~2 and suppose $C\in[\k_{2},\k_{1}+\k_{2}]$. We have
\begin{equation}
f^{*}(\k_{1}+\k_{2})\leq f^{*}(\k_{2})+M\label{eq:estimate_optsol_C1+C2}
\end{equation}
because~$f$ is subadditive and because $\k_{1}$ is the minimum
weight of all elements. Furthermore, we have
\begin{equation}
f\bigl(\algperm(\k_{2})\bigr)\geq f^{*}(\k_{2})-M\geq\rho-M\geq M\label{eq:estimate_algsol_C2}
\end{equation}
where the first inequality follows from subadditivity of~$f$ and
the fact that the prefix $\algperm(\k_{2})$ contains at least all
but one element from $\optsol(\k_{2})$. Combining \eqref{eq:estimate_optsol_C1+C2}
and \eqref{eq:estimate_algsol_C2}, we obtain
\[
	f(\algperm(\k_{2}))\geq f^{*}(\k_{1}+\k_{2})-2M\geq f^{*}(\k_{1}+\k_{2})-2f(\algperm(\k_{2})),
\]
i.e., by monotonicity,
\[
	f^{*}(C)\leq f^{*}(\k_{1}+\k_{2})\leq3f(\algperm(\k_{2}))\leq\rho f(\algperm(\k_{2}))\leq\rho f(\algperm(C)).
\]

Now consider phase $i\in\{3,\dots,N\}$ and $C\in\bigl(\sum_{j=1}^{i-1}\k_{j},\sum_{j=1}^{i}\k_{j}]$.
Note that, for $1\leq j\leq i\leq N-1$, $\k_{i}\geq\delta^{i-j}\k_{j}$
and hence
\[
\sum_{j=1}^{i-1}\frac{\k_{j}}{\k_{i}}\leq\sum_{j=1}^{i-1}\frac{1}{\delta^{i-j}}<\sum_{j=1}^{\infty}\frac{1}{\delta^{j}}=\frac{1}{\delta-1}<1,
\]
i.e., we have $\sum_{j=1}^{i-1}\k_{j}\leq\k_{i}\leq\sum_{j=1}^{i}\k_{j}$.
If $i=N$ and $\sum_{j=1}^{i-1}\k_{j}\geq\k_{i}$, we have nothing
left to show. Thus, suppose that we have $\sum_{j=1}^{N-1}\k_{j}\leq\k_{N}$.
Furthermore, if $i=N$ and $f^{*}(C_{N})<\rho f^{*}(C_{N-1})$, we
again have nothing to show as the prefix $\algperm(\sum_{j=1}^{i-1}\k_{j})\subseteq\algperm(C)$
has value at least $f^{*}(C_{N-1})$. Thus, assume that $f^{*}(C_{N})\geq\rho f^{*}(C_{N-1})$.
This implies that~\eqref{eq:value_in _phase_i} also holds for $i=N$.

Suppose $C\in\bigl(\sum_{j=1}^{i-1}\k_{j},\k_{i}\bigr)$. We will
show that in this case, the value of the solution $C_{i-1}$, which
is already added by the algorithm, is large enough to guarantee competitiveness.
If $\k_{i}>\delta\k_{i-1}$ holds, then $\k_{i}$ is the smallest
integer such that $f^{*}(\k_{i})\geq\rho f^{*}(\k_{i-1})$, i.e.,
using \eqref{eq:MO}, we have
\begin{align*}
f(\algperm(C))\geq f\bigl(\algperm\bigl(\textstyle\sum_{j=1}^{i-1}\k_{j}\bigr)\bigr) \geq f^{*}(\k_{i-1})>\frac{1}{\rho}f^{*}(C).
\end{align*}
For the case that $\k_{i}\leq\delta\k_{i-1}$, we distinguish between
two different cases:

\paragraph{\bf Case 1: $i=3$} Let $c:=\bigl(\frac{1}{\lambda\sqrt{M}}+\frac{1}{\lambda^{2}}\bigr)\delta\k_{2}$.
Note that $\bigl(\frac{1}{\lambda}+\frac{1}{\lambda^{2}}\bigr)\delta=\frac{\lambda^{2}}{\lambda+1}-1-\frac{1}{\delta}$
by definition of $\lambda$ and $\delta$ and thus
\begin{eqnarray}
c & \leq & \biggl(\frac{1}{\lambda}+\frac{1}{\lambda^{2}}\biggr)\delta\k_{2}\nonumber \\
 & = & \biggl(\frac{\lambda^{2}}{\lambda+1}-1-\frac{1}{\delta}\biggr)\k_{2}\nonumber \\
 & \leq & \frac{\lambda^{2}}{\frac{\lambda}{\sqrt{M}}+1}\k_{2}-\k_{2}-\k_{1}\nonumber \\
 & = & \frac{\lambda^{2}M}{\lambda\sqrt{M}+M}\k_{2}-\k_{2}-\k_{1}.\label{eq:estimate_for_j_in_phase_three}
\end{eqnarray}
 We will show that $\algperm(\k_{1}+\k_{2})$ is competitive up to
size $\k_{1}+\k_{2}+c$, and that $\algperm(\k_{1}+\k_{2}+c)$ is
competitive up to size $\k_{3}$. We have
\begin{eqnarray*}
f^{*}(\k_{1}+\k_{2}+c) & \overset{\eqref{eq:estimation_value-of-two-optimum-solutions}}{\leq} & \frac{\k_{1}+\k_{2}+c}{\k_{2}}\bigl(f^{*}(\k_{2})+M\bigr)\\
 & \overset{\eqref{eq:estimate_for_j_in_phase_three}}{\leq} & \frac{\k_{1}+\k_{2}+\Bigl(\frac{\lambda^{2}M}{\lambda\sqrt{M}+M}\k_{2}-\k_{2}-\k_{1}\Bigr)}{\k_{2}}\bigl(f^{*}(\k_{2})+M \bigr)\\
 & = & \frac{\lambda^{2}M}{\lambda\sqrt{M}+M}\biggl(1+\frac{M}{f^{*}(\k_{2})}\biggr)f^{*}(\k_{2})\\
 & \overset{\eqref{eq:value_in _phase_i}}{\leq} & \frac{\lambda^{2}M}{\lambda\sqrt{M}+M}\biggl(1+\frac{M}{\lambda\sqrt{M}}\biggr)f^{*}(\k_{2})\\
 & = & \lambda\sqrt{M}\frac{\lambda\sqrt{M}}{\lambda\sqrt{M}+M}\biggl(\frac{\lambda\sqrt{M}+M}{\lambda\sqrt{M}}\biggr)f^{*}(\k_{2})\\
 & \leq & \rho f^{*}(\k_{2})\\
 & \leq & \rho f(\algperm(\k_{1}+\k_{2})),
\end{eqnarray*}
where the last inequality follows from the fact that the algorithm
starts by packing $\optsol(\k_{1})$ and $\optsol(\k_{2})$ before
any other elements and needs capacity $\k_{1}+\k_{2}$ to assemble
both sets, i.e., $\optsol(\k_{2}) \subseteq\algperm(\k_{1}+\k_{2})$.

Since $\algscale$ adds the elements from $\optsol(\k_{3})$ after
those from $\optsol(\k_{1})$ and $\optsol(\k_{2})$, we have $\optsol(\k_{3},c) \subseteq\algperm(\k_{1}+\k_{2}+c)$,
and thus
\begin{eqnarray*}
f\bigl(\algperm(\k_{1}+\k_{2}+c)\bigr) & \overset{\eqref{eq:MO}}{\geq} & f(\optsol_{\k_{3},c})\\
 & \overset{\eqref{eq:estimate_value-of-subset}}{\geq} & \frac{c}{\k_{3}}f^{*}(\k_{3})-M\\
 & \geq & \biggl[\biggl(\frac{1}{\lambda\sqrt{M}}+\frac{1}{\lambda^{2}}\biggr)-\frac{M}{f^{*}(\k_{3})}\biggr]f^{*}(\k_{3})\\
 & \overset{\eqref{eq:value_in _phase_i}}{\geq} & \biggl(\frac{1}{\lambda\sqrt{M}}+\frac{1}{\lambda^{2}}-\frac{M}{\lambda^{2}M}\biggr)f^{*}(\k_{3})\\
 & = & \frac{1}{\lambda\sqrt{M}}f^{*}(\k_{3})\\
 & \geq & \frac{1}{\rho}f^{*}(\k_{3}),
\end{eqnarray*}
where for the third inequality we use that $\k_{3}\leq\delta\k_{2}$.
This, together with monotonicity of $f$, implies $f^{*}(C)\leq\rho f(\algperm(C))$
for all $C\in(\k_{1}+\k_{2},\k_{3}]$.

\paragraph{\bf Case 2: $i\geq4$}: Recall that $C\in\bigl(\sum_{j=1}^{i-1}\k_{j},\k_{i}\bigr)$.
We have
\begin{eqnarray*}
f^{*}(C) & \overset{\eqref{eq:MO}}{\leq} & f^{*}(\k_{i})\\
 & \overset{\eqref{eq:estimation_value-of-two-optimum-solutions}}{\leq} & \frac{\k_{i}}{\k_{i-1}}(f^{*}(\k_{i-1})+M)\\
 & \leq & \delta\biggl(1+\frac{M}{f^{*}(\k_{i-1})}\biggr)f^{*}(\k_{i-1})\\
 & \overset{\eqref{eq:value_in _phase_i}}{\leq} & \delta\biggl(1+\frac{M}{\lambda^{2}M}\biggr)f^{*}(\k_{i-1})\\
 & = & \frac{\lambda^{3}}{\lambda^{2}+1}\biggl(1+\frac{1}{\lambda^{2}}\biggr)f^{*}(\k_{i-1})\\
 & = & \lambda f^{*}(\k_{i-1})\\
 & \leq & \rho f(\algperm(C)),
\end{eqnarray*}
where for the third inequality we use $\k_{i}\leq\delta\k_{i-1}$.
Thus, also in this case, we find $f^{*}(C)\leq\rho f(\algperm(C))$
for all $C\in\bigl(\sum_{j=1}^{i-1}\k_{j},\k_{i}\bigr)$.

Now, consider $C\in\bigl[\k_{i},\sum_{j=1}^{i}\k_{j}\bigr]$ with
$i<N$. Up to this budget, the algorithm had a capacity of $C-\sum_{j=1}^{i-1}\k_{j}>C-\k_{i}\geq0$
to pack elements from $\optsol(\k_{i})$, i.e., $\optsol(\k_{i},C-\sum_{j=1}^{i-1}\k_{j})\subseteq\algperm(C)$.
We will show that the value of this set is large enough to guarantee
competitiveness in this case. We have
{
\allowdisplaybreaks
\begin{eqnarray*}
f\bigl(\algperm(C)\bigr) & \overset{\eqref{eq:MO}}{\geq} & f\Bigl(\optsol\bigl(\k_{i},C-\textstyle\sum_{j=1}^{i-1}\k_{j}\bigr)\Bigr)\\
 & \overset{\eqref{eq:estimate_value-of-subset}}{\geq} & \frac{C-\sum_{j=1}^{i-1}\k_{j}}{\k_{i}}f^{*}(\k_{i})-M\\
 & \overset{\eqref{eq:estimation_value-of-two-optimum-solutions}}{\geq} & \frac{C-\sum_{j=1}^{i-1}\k_{j}}{\k_{i}}\biggl(\frac{\k_{i}}{C}f^{*}(C)-M\biggr)-M\\
 & = & \biggl(\frac{C-\sum_{j=1}^{i-1}\k_{j}}{C}-\frac{C-\sum_{j=1}^{i-1}\k_{j}}{\k_{i}}\cdot\frac{M}{f^{*}(C)}-\frac{M}{f^{*}(C)}\biggr)f^{*}(C)\\
 & \overset{}{\geq} & \biggl(1-\sum_{j=1}^{i-1}\frac{\k_{j}}{\k_{i}}-1\cdot\frac{M}{f^{*}(C)}-\frac{M}{f^{*}(C)}\biggr)f^{*}(C)\\
 & \overset{\eqref{eq:value_in _phase_i}}{\geq} & \biggl(1-\sum_{j=1}^{i-1}\frac{1}{\delta^{i-j}}-\frac{2M}{\rho^{i-1}}\biggr)f^{*}(C)\\
 & \geq & \biggl(1-\sum_{j=1}^{\infty}\frac{1}{\delta^{j}}-\frac{2M}{\rho^{i-1}}\biggr)f^{*}(C)\\
 & \geq & \biggl[1-\biggl(\frac{1}{1-\delta^{-1}}-1\biggr)-\frac{2M}{\lambda^{2}M}\biggr]f^{*}(C)\\
 & \geq & 0.319\cdot f^{*}(C)\\
 & \geq & \frac{1}{\rho}f^{*}(C),
\end{eqnarray*}
}
where for the fourth inequality we use $\k_{i}\leq C\leq\sum_{j=1}^{i}\k_{j}$ and for the fifth inequality we use $\k_{j+1}\geq\delta\k_{j}$.
\end{proof}

For $1$-bounded objectives, Theorem~\ref{thm:upper_bound_M} immediately
yields the following.
\begin{corollary}
If $M=1$, the incremental solution computed by $\algscale$ is $3.2924$-competitive.
\end{corollary}

\section{Lower bound\label{sec:lb}}

In this section, we show the second part of \Cref{thm:upper_and_lower_bound_for_M-bounded},
i.e., we give a lower bound on the competitive ratio for the incremental
optimization problem~(\ref{eq:generic-opt}) with monotone, $M$-bounded,
and fractionally subadditive objectives, and we show a lower bound
for the special case with $M=1$.

\begin{theorem}
For monotone, $M$-bounded, and fractionally subadditive objectives,
the knapsack problem~(\ref{eq:generic-opt}) does not admit a $\rho$-competitive
incremental solution for $\rho<M$.
\end{theorem}

\begin{proof}
Consider the set $\gset=\{e_{1},e_{2}\}$ with weights $w(e_{i})=i$
for $i\in\{1,2\}$ and the values $v(e_{1})=1$ and $v(e_{2})=M$.
We define the objective $f(S):=\sum_{e\in S}v(e)$ for all $S\subseteq\gset$.
It is easy to see that $f$ is monotone, $M$-bounded and modular
and thus fractionally subadditive.

Consider some competitive algorithm with competitive ratio $\rho\geq0$
for the knapsack problem~(\ref{eq:generic-opt}). In order to be
competitive for capacity $1$, the algorithm has to add element $e_{1}$
first. Thus, the solution of the algorithm of size $2$ cannot contain
element~$e_{2}$, i.e., the value of the solution of capacity $2$
given by the algorithm has value $1$. The optimal solution of capacity
$2$ has value $M$, and thus $\rho\geq M$.
\end{proof}

We proceed to give a stronger lower bound for $M\in[1,\varphi+1)$ with $\varphi=\frac{1}{2}(1+\sqrt{5})$.
To this end, for a value
$n\in\N$, consider an instance~$I$ with $\sum_{i=1}^{n}i=\frac{1}{2}n(n+1)$
elements partitioned into sets $E_{1},E_{2},\dots,E_{n}$ such that
$|E_{i}|=i$ for all $i\in\left\{ 1,\dots,n\right\} $. We set 
\begin{align}
	\label{eq:definition_of_lower_bound_objective_function}
f(S)=\max_{i\in\{1,\dots,n\}} |S\cap E_{i}| \quad \text{for all $S\subseteq\gset$.}
\end{align}

The elements' weights are defined as $w(e)=\base+i!$
for all $e\in E_{i}$ with base weight $\base= (n+2)!$.

We note that the problem instance~$I$
is built in such a way that the elements in
all sets $\gset_{1},\dots,\gset_{n}$ have roughly the same relative weight
because $\base$ is very large. We show first, for a given capacity $C\in\N$,
the number of elements that can be packed without exceeding this capacity can vary
by at most 1, regardless of which elements are packed. Yet, the weights
of elements in $\gset_{i}$ increase quickly enough with increasing~$i$ such that, for capacity $C=i(\base+i!)$ it is only possible to pack
$i$ elements if all $i$ elements are from the set $\gset_{1}\cup\dots\cup\gset_{i}$.

\begin{proposition}
	\label{prop:number_of_elements_in_lower_bound_solution}
	Let $\pi$ be a solution to the instance defined above. Consider capacity $C=i(b+i!)$ for $i \in \{1,\dots, \frac{1}{2}n(n+1)\}$.
	%Then we have
	%	$\left|\left|\optsol(C)\right|-\left|\pi(C)\right|\right| \leq 1$
	%	for all capacities $C\geq 0$.
	%	For two solutions $\pi$ and $\pi'$ to the instance defined above we
	%	have $\left|\left|\pi(C)\right|-\left|\pi'(C)\right|\right| \leq 1$
	%	for all capacities $C\geq 0$.
	%	For $C= i(\base+i!)$, we have $|\pi(C)|=i$ if and only if
	%	$\{e_{\pi(1)},\dots,e_{\pi(i)}\} \in E_1 \cup \dots \cup E_i$.
	Then, we have
	\[
	|\pi(C)|=
	\begin{cases}
		i &\text{if } \{e_{\pi(1)},\dots,e_{\pi(i)}\} \subseteq E_1 \cup \dots \cup E_i, \\
		i-1 &\text{else}
	\end{cases}
	\]
	with $E_j= \emptyset$ for $j>n$.
\end{proposition}
\begin{proof}
	First we observe that $|\pi(C)| \in \{i,i-1\}$: Assume $|\pi(C)|< i-1$. Then, we have
	\begin{align*}
		C-w(\pi(C)) &\geq i(b+i!)- (i-2)(b+n!) \geq 2b-in! \\
		&\geq b+ (n+2)!- n(n+1)n! = b+2(n+1)! \geq \max_{e \in E} w(e),
		%		b+n! &= \max_{e \in E} w(e) > C- w(\pi(C)) = i(b+i!)- |\pi(C)|(b+n!) \\
		%		&\geq (i-1-|\pi(C)|)b + (n+2)!- |\pi(C)|n! \geq (i-1-|\pi(C)|)b+n!
	\end{align*}
	contradicting that $\pi(C)$ is the maximum prefix of $\pi$ which fits into capacity~$C$.
	Assume $|\pi(C)|> i$. Then, we have $w(\pi(C)) > (i+1)b > ib+(i+1)! > i(b+i!)$, which contradicts $w(\pi(C))\leq C$.
	
	If $\pi(C)$ contains $e \in E_j$ with $j>i$, we have
	\begin{align*}
		C-w(e) &= i(b+i!)-(b+j!) = (i-1)b+ii!-j! \\
		&< (i-1)b < (i-1) \min_{e \in E} w(e).
	\end{align*}
	So $\pi(C) \setminus e$ contains at most $i-2$ elements and $|\pi(C)|\leq i-1$. If the first~$i$ elements in $\pi$ are in the sets $E_1, \dots E_i$, we have $w(\{e_{\pi(1)}, \dots, e_{\pi(i)}\}) \leq i(b+i!) = C$.
	Therefore, we have $|\pi(C)|\geq i$.
\end{proof}

We say an incremental solution to the problem instance~$I$ given above
\emph{is represented by} a set of numbers $A=\{a_1, \dots, a_\ell\}$ with
$a_i < a_{i+1}$ and $a_\ell=n$ if the solution adds first all elements
from~$E_{a_1}$, then from $E_{a_2}$, and so on until adding all elements
from $E_{a_\ell}$. Only afterwards all remaining elements are added in an
arbitrary order. Note, that elements added after the last element of
$E_n$ in any solution do not influence the objective value for any
capacity, since when they are added the incremental solution has already
reached the maximum value of $n$. %Let $m_{\pi}$ denote the position of
%the last element from $E_n$ in $\pi$ for any solution $\pi$, i.e. $e_{\pi(m_{\pi})} \in E_n$ and $e_{\pi(h)} \notin E_n$ for all indices $h>m_{\pi}$.
First we will observe, that every incremental solution of problem instance~$I$ can be
transformed into a solution that can be represented by a set $\{a_1,...,a_\ell\}$
without decreasing the objective value for any capacity.

\begin{lemma}
	\label{lem:ordered_lower_bound_solution}
	For every incremental solution $\pi$ there is a set $A=\{a_1, \dots, a_\ell\}$ with $a_i < a_{i+1}$ and $a_{\ell}=n$ representing an incremental solution with objective value at least $f(\pi(C))$ for all capacities~$C\geq0$.
\end{lemma}
\begin{proof}
	First, we show that there is a solution $\pi'$, whose objective value is at least $f(\pi(C))$ for every capacity, such that, for all $i\in\{1,...,n-1\}$, if at least one element from the set $E_i$ is added before the last element from $E_n$ is added, then this is true for all elements from $E_i$.
	Furthermore, if $h\in\N$ is the index of the last element from~$E_i$ in $\pi'$, we have $f(\{e_{\pi'(1)},...,e_{\pi'(h-1)}\}) + 1 = i = f(\{e_{\pi'(1)},...,e_{\pi'(h)}\})$.
	We will do this by altering the solution~$\pi$ to obtain the desired solution~$\pi'$.
	
	Fix some $i\in\N$ such that at least one element from the set $E_i$ is added before the last element from $E_n$ is added in the solution~$\pi$.
	Let $j\in\N$, $j\leq i$ be the largest number such that, when the $j$-th element from the set $E_i$ is added, the value of the solution increases from $j-1$ to~$j$.
	If this does not exist, we set $j=0$.
	If $j=i$, $E_i$ is completely added before the last element from $E_n$ and when the last element from $E_i$ is added to the incremental solution its value increases by one to the value of $i$.
	Thus, suppose that $j<i$.
	All elements from $E_i$ that are added after the $j$-th element do not increase the value of the solution and can thus be moved to the end of the whole order $\pi$.
	Since now, there are only~$j$ elements from the set $E_i$ added before the last element from the set $E_n$ is added, it makes sense to add the elements from the set~$E_j$ instead of these~$j$ elements, as they have a smaller weight (if they are not already added).
	We can then move the~$j$ elements from $E_i$ to the end of the order.
	After all these changes the solution obtains all values at least as fast as before.
	By doing this for all $i\in\{1,...,n-1\}$, we obtain the desired solution~$\pi'$.
	
	Now, we will show that we can reorder the elements in the solution~$\pi'$ such that the elements that are added before the last element from $E_n$ is added are ordered by the index of the set they belong to.
	Consider any two sets $E_i$, $E_j$, $i<j$ that are added before the last element from set $E_n$ is added.
	Recall that, when the last element from~$E_i$ is added, the value of the solution is~$i$.
	This implies that at that point at most $i-1$ elements from the set~$E_j$ are added.
	Thus, swapping the elements of~$E_i$ and~$E_j$ in the order $\pi'$ until all elements from $E_i$ are added before the elements from~$E_j$, does not decrease the value of the solution for any capacity.
	By doing this for all pairs $(i,j)$, we obtain a solution that can be represented by a set $A=\{a_1,...,a_\ell\}$.
\end{proof}	

Utilizing the properties of the weights we mentioned before, we can
find a collection of conditions which are necessary and sufficient
for a set of numbers $\{a_{1},\dots,a_{\ell}\}$ to represent a $\rho$-competitive
solution for the problem instance~$I$.
In the following, we denote by $\ell'$ the index with $\rho a_{\ell'}<n$ and $\rho a_{\ell'+1} \geq n$ and set $z_{i}:=\sum_{j=1}^{i}a_{j}$. The index $\ell'$ is needed because all indices $i>\ell'$ satisfy $\rho a_i \geq n$, i.e., after a solution has added the set $E_i$, it is $\rho$-competitive for all capacities.
\begin{lemma}
\label{lem:necessary_and_sufficient_condition_for_competitiveness_of_set}Let
$A=\{a_{1},\dots,a_{\ell}\}$ with $a_{i}<a_{i+1}$ and $a_{\ell}=n$ be a set of numbers
that represents an incremental solution for instance~$I$. Then, the incremental solution
is $\rho$-competitive if and only if the following three conditions
are satisfied:
\begin{enumerate}[label=\textrm{(\roman*)}, ref=\textit{(\roman*)}]
	\item \label{it:numbers:1} $a_{1}=1$,
	
	\item the following two conditions hold for all $i\in\{1,\dots,\ell'\}$:
	\begin{enumerate}[label=(ii\alph*), ref=\textit{(ii\alph*)}]
	\item \label{it:numbers:2a} $z_{i}+a_{i}\leq\lfloor\rho a_{i}\rfloor$ if $a_{i+1}\leq\lfloor\rho a_{i}\rfloor+1$,
	\item \label{it:numbers:2b} $z_{i}+a_{i}+1\leq\lfloor\rho a_{i}\rfloor$ if $a_{i+1}>\lfloor\rho a_{i}\rfloor+1$.
	\end{enumerate}
\end{enumerate}
\end{lemma}

\begin{proof}
We first show that for a $\rho$-competitive incremental solution represented by a set of numbers~$A$ conditions~\ref{it:numbers:1}, \ref{it:numbers:2a} and \ref{it:numbers:2b} have to be satisfied.

If $a_{1}\neq1$, the incremental solution is not competitive for capacity
$C=\base+1$, i.e.,~\ref{it:numbers:1} must hold.

Consider capacity $C=(\lfloor\rho a_{i}\rfloor+1)(\base+(\lfloor\rho
a_{i}\rfloor+1)!)$ for $i \in \{1,\dots,l'\}$. The optimum solution of
capacity $C$ is $\optsol(C) = \gset_{\lfloor\rho a_{i}\rfloor+1}$
and has value $f^*(C) = \lfloor\rho a_{i}\rfloor+1$. The value of
the incremental solution is at least $a_{i}+1$, since $\frac{1}{\rho}(\lfloor\rho
a_{i}\rfloor+1)>a_{i}$. Thus, the incremental solution of capacity~$C$ contains at least $a_{i}+1$
elements from $a_{i+1}$.

If $a_{i+1} \leq \lfloor\rho a_{i}\rfloor+1$, the incremental solution contains
$\lfloor\rho a_{i}\rfloor+1-z_i$ from $a_{i+1}$
by~\Cref{prop:number_of_elements_in_lower_bound_solution} at capacity~$C$. Thus, we have
$a_{i}+1 \leq \lfloor\rho a_{i}\rfloor+1-z_i$ which
implies~\ref{it:numbers:2a}. If $a_{i+1} > \lfloor\rho a_{i}\rfloor+1$,
the incremental solution contains $\lfloor\rho a_{i}\rfloor-z_i$ from~$a_{i+1}$
by~\Cref{prop:number_of_elements_in_lower_bound_solution}. Thus, we have
$a_{i}+1 \leq \lfloor\rho a_{i}\rfloor-z_i$ which implies
\ref{it:numbers:2b}.

We proceed to show that, conversely, an incremental solution represented by a set of numbers~$A$ satisfying conditions \ref{it:numbers:1}, \ref{it:numbers:2a} and \ref{it:numbers:2b} is $\rho$-competitive. To this end, fix an arbitrary incremental solution with these properties.
Since all elements have integer weight, it suffices to show $\rho$-competitiveness
for all capacities $C\in\N$.

For capacities $C\in\{1,\dots,\base+1\}$, the incremental solution is $\rho$-competitive
because~$\base+1$ is the smallest weight of all elements and $a_{1}=1$
by \ref{it:numbers:1}, i.e., the element of smallest weight is added first.

Let $i\in\{1,\dots,\ell'\}$.
We will show that, for all capacities
in 
\begin{align*}
\bigl\{z_{i}(\base+z_{i}!)+1,\dots,z_{i+1}(\base+z_{i+1}!)\bigr\},
\end{align*}
the incremental solution is $\rho$-competitive. For all capacities
\begin{align}
\label{eq:c-values}
C \in \bigl\{z_{i}(\base+z_{i}!)+1,\dots,(\lfloor\rho a_{i}\rfloor+1)(\base+(\lfloor\rho a_{i}\rfloor+1)!)-1 \bigr\},
\end{align}
the value of the optimum solution is at most $\lfloor\rho a_{i}\rfloor$
while the value of the incremental solution is at least~$a_{i}$
because it contains at least all elements from $\gset_{1},\dots,\gset_{a_{i}}$
at capacity $z_{i}(\base+z_{i}!)$ by \Cref{prop:number_of_elements_in_lower_bound_solution}.
%\[
%\sum_{j=1}^{i}a_{j}(\base+a_{j}!)=z_{i}\base+\sum_{j=1}^{i}a_{j}a_{j}!\leq z_{i}\base+z_{i}z_{i}!\leq z_{i}(\base+z_{i}!)+1.
%\]
Thus, the incremental solution is $\rho$-competitive for all values~$C$ as in \eqref{eq:c-values}.
Next, suppose that
\begin{align*}
C \in \bigl\{(\lfloor\rho a_{i}\rfloor+1)(\base+(\lfloor\rho a_{i}\rfloor+1)!),\dots,z_{i+1}(\base+z_{i+1}!) \bigr\}.
\end{align*}
Let $a^{*}\in\{\lfloor\rho a_{i}\rfloor+1,\dots,z_{i+1}\}$
be the value of the optimum solution of capacity~$C$. This implies
that $C\geq a^{*}(\base+a^{*}!)$. We consider two cases.

\paragraph{\bf Case 1: $a_{i+1}\leq\lfloor\rho a_{i}\rfloor+1$} By \ref{it:numbers:2a}, we have $z_{i}+a_{i}\leq\lfloor\rho a_{i}\rfloor$
and thus
\begin{align}
a^{*}-z_{i}  \geq  a_{i}+a^{*}-\lfloor\rho a_{i}\rfloor 
  \geq  \frac{1}{\rho}\lfloor\rho a_{i}\rfloor+\frac{1}{\rho}\bigl(a^{*}-\lfloor\rho a_{i}\rfloor\bigr) 
  =  \frac{1}{\rho}a^{*}.\label{eq:a-star_minus_zi_competitive}
\end{align}
%We have
%\begin{align*}
% &  \Biggl(\sum_{j=1}^{i}a_{j}(\base+a_{j}!)\Biggr)+(a^{*}-z_{i})(\base+a_{i+1}!)\\
% & \quad =  a^{*}\base+\Biggl(\sum_{j=1}^{i}a_{j}a_{j}!\Biggr)+(a^{*}-z_{i})a_{i+1}!\\
% & \quad <  a^{*}\base+a^{*}a_{i+1}!\\
% & \quad \leq  a^{*}(\base+a^{*}!)\\
% & \quad \leq  C,
%\end{align*}
%where for the first inequality we use that $a_{j}<a_{j+1}$ and for the second inequality we use that $a_{i+1}\leq\lfloor\rho a_{i}\rfloor+1\leq a^{*}$.
By \Cref{prop:number_of_elements_in_lower_bound_solution}
%Thus,
the incremental solution contains at least $a^{*}-z_{i}$ elements from the
set $\gset_{a_{i+1}}$ at capacity~$C$ since $a_{i+1}\leq\lfloor\rho a_{i}\rfloor+1 \leq a^*$. This means that, by~(\ref{eq:a-star_minus_zi_competitive}),
the incremental solution of capacity~$C$ has value at least
$a^{*}-z_{i}\geq\frac{1}{\rho}a^{*}$, i.e., the solution is $\rho$-competitive
for capacity $C$.

\paragraph{\bf Case 2: $a_{i+1}>\lfloor\rho a_{i}\rfloor+1$} By \ref{it:numbers:2b},
we have $z_{i}+a_{i}+1\leq\lfloor\rho a_{i}\rfloor$ and thus
\begin{align}
a^{*}-z_{i}-1  \geq  a_{i}+a^{*}-\lfloor\rho a_{i}\rfloor 
  \geq  \frac{1}{\rho}\lfloor\rho a_{i}\rfloor+\frac{1}{\rho}\bigl(a^{*}-\lfloor\rho a_{i}\rfloor\bigr)
  =  \frac{1}{\rho}a^{*}.\label{eq:a-star_minus_zi_minus_one_competitive}
\end{align}
%We have
%\begin{align*}
% & \Biggl(\sum_{j=1}^{i}a_{j}(\base+a_{j}!)\Biggr)+(a^{*}-z_{i}-1)(\base+a_{i+1}!)\\
% & \quad =  (a^{*}-1)\base+\Biggl(\sum_{j=1}^{i}a_{j}a_{j}!\Biggr)+(a^{*}-z_{i}-1)a_{i+1}!\\
% & \quad \leq  (a^{*}-1)\base+i\cdot n\cdot n!+n\cdot n!\\
% & \quad \leq  (a^{*}-1)\base+(n+2)!\\
% & \quad =  (a^{*}-1)\base+\base\\
% & \quad \leq  a^{*}(\base+a^{*}!)\\
% & \quad \leq  C,
%\end{align*}
%where for the second inequality we use that $i\leq n$, and for the penultimate inequality we use that $a^{*}!\geq 0$.
%Thus,
By \Cref{prop:number_of_elements_in_lower_bound_solution} the incremental solution contains at least $a^{*}-z_{i}-1$ elements from
the set $\gset_{a_{i+1}}$ at capacity~$C$. This means that, by~(\ref{eq:a-star_minus_zi_minus_one_competitive}),
the incremental solution of capacity~$C$ has value at least
$a^{*}-z_{i}-1\geq\frac{1}{\rho}a^{*}$, i.e., the solution is
$\rho$-competitive for capacity~$C$.

We conclude that for every capacity~$C \in \bigl\{1,\dots,z_{\ell'+1}(\base+z_{\ell'+1}!)\bigr\}$ the incremental solution is $\rho$-competitive. For all capacities $C > z_{\ell'+1}(\base+z_{\ell'+1}!)$,
the value of the incremental solution is at least~$a_{\ell'+1}$, while
the optimum solution has value at most~$n$. By definition of $\ell'$ we have $\rho a_{\ell'+1} \geq n$. Therefore, the incremental solution
is $\rho$-competitive.
\end{proof}

In the following we will % be able to
show that, for $2\leq\rho\leq\varphi+1$ and given some
set of numbers $\{a_{1},\dots,a_{i}\}$, every algorithm is forced to choose
$a_{i+1}\leq\lfloor\rho a_{i}\rfloor+1$ to be $\rho$-competitive for capacity
$\lfloor\rho a_{i}\rfloor+1$.

Applying condition \ref{it:numbers:2b} from \Cref{lem:necessary_and_sufficient_condition_for_competitiveness_of_set}
yields the following.
\begin{corollary}
\label{cor:condition_for_following_number}If a set of numbers $\{a_{1},\dots,a_{\ell}\}$
represents a $\rho$-competitive incremental solution and $z_{i}+a_{i}=\lfloor\rho a_{i}\rfloor$
for some $i\in\{1,\dots,\ell'\}$, then
\begin{align}
\label{eq:condition_for_following_number}	
a_{i+1}\leq\lfloor\rho a_{i}\rfloor+1.
\end{align}
\end{corollary}

\begin{proposition}
\label{prop:small_rho_bounded_number-size}Let $\rho\in[2,\varphi+1]$,
and let $A=\{a_{1},\dots,a_{\ell}\}$ with $a_{i}<a_{i+1}$ be a set
of numbers that represents an incremental solution. If the incremental solution
is $\rho$-competitive, then $a_{i+1}\leq\lfloor\rho a_{i}\rfloor+1$
for all $i\in\{1,\dots,\ell'\}$.
\end{proposition}

\begin{proof}
By \Cref{cor:condition_for_following_number}, it suffices
to show that we have $z_{i}+a_{i}=\lfloor\rho a_{i}\rfloor$ for all $i\in\{1,\dots,\ell'\}$.
We will prove this by induction. For $i=1$, we have
\begin{align*}
z_{1}+a_{1}=1+1=2=\lfloor\rho\rfloor=\lfloor\rho a_{1}\rfloor,
\end{align*}
where we use the fact that $a_{1}=1$ by \Cref{lem:necessary_and_sufficient_condition_for_competitiveness_of_set}\ref{it:numbers:1}.

Suppose that
\begin{equation}
z_{i}+a_{i}=\lfloor\rho a_{i}\rfloor\label{eq:induction_hypothesis}
\end{equation}
holds for some $i\in\{1,\dots,\ell'-1\}$. By \Cref{lem:necessary_and_sufficient_condition_for_competitiveness_of_set}\ref{it:numbers:2a} and \ref{it:numbers:2b} and the $\rho$-com\-pe\-ti\-ti\-ve\-ness of the incremental solution,
we have $z_{i+1}+a_{i+1}\leq\lfloor\rho a_{i+1}\rfloor$. Thus, we
only have to show that
\begin{equation}
z_{i+1}+a_{i+1}\geq\lfloor\rho a_{i+1}\rfloor.\label{eq:estimate_set-numbers_size_ind_to-show}
\end{equation}
To prove this, we first calculate for $\rho \in (2,\varphi+1]$:
\begin{multline}
%\frac{\rho-2}{(3-\rho)(\rho-1)}  =  \frac{\rho-2}{-(\rho-2)^{2}+1}
% =  \frac{1}{\frac{1}{\rho-2}-(\rho-2)} \\
% \leq  \frac{1}{\frac{1}{\varphi-1}-(\varphi-1)}
%  =  \frac{1}{\varphi-(\varphi-1)} =  1, \label{eq:ineq_ind_1}
\frac{(3-\rho)(\rho-1)}{\rho-2}  =  \frac{-(\rho-2)^{2}+1}{\rho-2}
=  \frac{1}{\rho-2}-(\rho-2) \\
\geq  \frac{1}{\varphi-1}-(\varphi-1)
=  \varphi-(\varphi-1) =  1, \label{eq:ineq_ind_1}
\end{multline}
where for the inequality we use $\rho\leq\varphi+1$. If $\rho=2$ we can directly calculate $\rho-2 = 0 \leq 1 = (3-\rho)(\rho-1)$.
We then obtain
\begin{eqnarray}
(3-\rho)\lfloor(\rho-1)a_{i}\rfloor+1 & > & (3-\rho)((\rho-1)a_{i}-1)+1\nonumber \\
 & = & (3-\rho)(\rho-1)a_{i}+\rho-2\nonumber \\
 & \overset{\eqref{eq:ineq_ind_1}}{\geq} & (\rho-2)a_{i}+\rho-2\nonumber \\
 & = & (\rho-2)(a_{i}+1).\label{eq:ineq_ind_2}
\end{eqnarray}
Utilizing this inequality, we have
\begin{eqnarray}
\lfloor(\rho-2)(\lfloor\rho a_{i}\rfloor+1)\rfloor & = & \lfloor(\rho-2)(\lfloor(\rho-1)a_{i}\rfloor+a_{i}+1)\rfloor\nonumber \\
 & = & \lfloor\lfloor(\rho-1)a_{i}\rfloor+(\rho-3)\lfloor(\rho-1)a_{i}\rfloor+(\rho-2)(a_{i}+1)\rfloor\nonumber \\
 & = & \lfloor(\rho-1)a_{i}\rfloor+\lfloor(\rho-3)\lfloor(\rho-1)a_{i}\rfloor+(\rho-2)(a_{i}+1)\rfloor\nonumber \\
 & \overset{\eqref{eq:ineq_ind_2}}{<} & \lfloor(\rho-1)a_{i}\rfloor+\lfloor1\rfloor,\label{eq:ineq_ind_3}
\end{eqnarray}
where for the third equation we use that $\lfloor(\rho-1)a_{i}\rfloor\in\N$.
Because both sides of this inequality are in $\N$, we have
\begin{eqnarray*}
\lfloor\rho a_{i+1}\rfloor & = & \lfloor(\rho-2)a_{i+1}\rfloor+2a_{i+1}\\
 & \overset{\eqref{eq:condition_for_following_number}}{\leq} & \lfloor(\rho-2)(\lfloor\rho a_{i}\rfloor+1)\rfloor+2a_{i+1}\\
 & \overset{\eqref{eq:ineq_ind_3}}{\leq} & \lfloor(\rho-1)a_{i}\rfloor+2a_{i+1}\\
 & = & \lfloor\rho a_{i}\rfloor-a_{i}+2a_{i+1}\\
 & \overset{\eqref{eq:induction_hypothesis}}{=} & z_{i}+2a_{i+1}\\
 & = & z_{i+1}+a_{i+1},
\end{eqnarray*}
i.e., (\ref{eq:estimate_set-numbers_size_ind_to-show}) holds. By
induction $z_{i}+a_{i}=\lfloor\rho a_{i}\rfloor$ follows for all
$i\in\{1,\dots,\ell'\}$, and therefore, by \Cref{cor:condition_for_following_number},
the proposition holds.
\end{proof}

\begin{theorem}
\label{thm:lower_bound_phi+1}
For $\rho<\varphi+1$, there is no $\rho$-competitive algorithm for problem instance~$I$ with sufficiently large~$n\in\N$.
\end{theorem}

\begin{proof}
Suppose, for $\rho<2$, there was a $\rho$-competitive incremental solution
represented by the set of numbers $\{a_{1},\dots,a_{\ell}\}$. Without
loss of generality we can assume that $a_{i}<a_{i+1}$ for all $i\in\{1,\dots,\ell-1\}$.
Yet, \Cref{lem:necessary_and_sufficient_condition_for_competitiveness_of_set}\ref{it:numbers:1}, \ref{it:numbers:2a} and \ref{it:numbers:2b} imply that
\begin{align*}
2=z_{1}+a_{1}\leq\lfloor\rho a_{1}\rfloor=1
\end{align*}
which is a contradiction, i.e., for $\rho<2$, there is no $\rho$-competitive
incremental solution.

Next, suppose that for $\rho\in[2,\varphi+1)$ there was a $\rho$-competitive
incremental solution.
Let the number of disjoint sets $n\in\N$ in the instance be sufficiently large, and
let $\{a_{1},\dots,a_{\ell}\}$ be the set of numbers representing a $\rho$-competitive incremental solution.
Without loss of generality, we can assume that $a_{i+1}>a_{i}$ for
all $i\in\{1,\dots,\ell-1\}$. By \Cref{lem:necessary_and_sufficient_condition_for_competitiveness_of_set}
and \Cref{prop:small_rho_bounded_number-size}, we know that the following conditions are satisfied:
\begin{enumerate}[label=\textit{(\roman*)}, ref=\textit{(\roman*)}]
\item \label{it:condition:1} $a_{1}=1$,
\item \label{it:condition:2} $z_{i}+a_{i}\leq\lfloor\rho a_{i}\rfloor$ for all $i\in\{1,\dots,\ell'\}$,
%\item \label{it:condition:3} $a_{\ell}\geq\frac{1}{\rho}n$,
\item \label{it:condition:4} $a_{i+1}\leq\lfloor\rho a_{i}\rfloor+1$ for all $i\in\{1,\dots,\ell'\}$.
\end{enumerate}
For $1\leq j \leq i\leq\ell'$, from \ref{it:condition:4} it follows that
\[
a_{j}\geq\frac{1}{\rho}\lfloor\rho a_{j}\rfloor\overset{\ref{it:condition:4}}{\geq}\frac{1}{\rho}(a_{j+1}-1)\geq\frac{1}{\rho}\biggl[\frac{1}{\rho}\biggl(a_{j+2}-1\biggr)-1\biggr]\geq\dots\geq\frac{1}{\rho^{i-j}}a_{i}-\sum_{k=1}^{i-j}\frac{1}{\rho^{k}}.
\]
This implies
\begin{eqnarray}
z_{i} & = & \sum_{j=1}^{i}a_{j}\nonumber \\
 & \geq & \sum_{j=1}^{i}\Biggl(\frac{1}{\rho^{i-j}}a_{i}-\sum_{k=1}^{i-j}\frac{1}{\rho^{k}}\Biggr)\nonumber \\
 & = & \Biggl(\sum_{j=0}^{i-1}\frac{1}{\rho^{j}}\Biggr)a_{i}-\sum_{j=1}^{i}\sum_{k=1}^{i-j}\frac{1}{\rho^{k}}\nonumber \\
 & = & \frac{1-\rho^{-i}}{1-\rho^{-1}}a_{i}-\sum_{j=1}^{i}\biggl(\frac{1-\rho^{j-i-1}}{1-\rho^{-1}}-1\biggr)\nonumber \\
 & \geq & \frac{1-\rho^{-i}}{1-\rho^{-1}}a_{i}-i\frac{1}{1-\rho^{-1}}\nonumber \\
 & = & \frac{1}{1-\rho^{-1}}\bigl((1-\rho^{-i})a_{i}-i\bigr).\label{eq:sum-of-numbers_estimate}
\end{eqnarray}
For $i\in\{2,\dots,\ell'\}$, we obtain
\begin{eqnarray}
\rho & \geq & \frac{1}{a_{i}}\lfloor\rho a_{i}\rfloor\nonumber \\
 & \overset{\ref{it:condition:2}}{\geq} & \frac{1}{a_{i}}(z_{i}+a_{i})\nonumber \\
 & \overset{\eqref{eq:sum-of-numbers_estimate}}{\geq} & \frac{1}{a_{i}}\cdot\frac{1}{1-\rho^{-1}}\bigl((1-\rho^{-i})a_{i}-i\bigr)+1\nonumber \\
 & = & \frac{1}{1-\rho^{-1}}\biggl(1-\rho^{-i}-\frac{i}{a_{i}}\biggr)+1.\label{eq:beginning_contradiction}
\end{eqnarray}
Observe that $a_{j+1}>a_j$ for all $j\in\{1,...,\ell-1\}$ implies $a_j\geq j$ for all $j\in\{1,...,\ell\}$. It follows that
\begin{eqnarray*}
a_{i} & \geq & \frac{1}{\rho-1}\bigl(\lfloor\rho a_{i}\rfloor-a_{i}\bigr)\\
 & \overset{\ref{it:condition:2}}{\geq} & \frac{1}{\rho-1}z_{i}\\
 & \overset{a_{j}\geq j}{\geq} & \frac{1}{\rho-1}\cdot\frac{i(i+1)}{2},
\end{eqnarray*}
which implies that
\begin{equation}
\frac{i}{a_{i}}\leq\frac{2(\rho-1)}{i+1}.\label{eq:estimate_i:a_i}
\end{equation}
By definition of $\ell'$ and by \Cref{prop:small_rho_bounded_number-size},
$\ell'$ increases when $n$ is increased sufficiently. Thus, for every
$\varepsilon>0$, there exists some $n\in\N$ such that
\begin{equation}
\frac{\ell'}{a_{\ell'}}\overset{\eqref{eq:estimate_i:a_i}}{\leq}\frac{2(\rho-1)}{\ell'+1}\leq\frac{\varepsilon}{2}\label{eq:epsilon_estimate_1}
\end{equation}
and
\begin{equation}
\rho^{-\ell'}\leq\frac{\varepsilon}{2}.\label{eq:epsilon_estimate_2}
\end{equation}
Since we chose $n$ sufficiently large, we can assume that this holds.
By choosing $\varepsilon=1-\frac{\rho-1}{\rho}\varphi$ we see that $\varepsilon > 0$ since $\rho<\varphi+1$, and
we obtain
\begin{eqnarray*}
\rho & \overset{\eqref{eq:beginning_contradiction}}{\geq} & \frac{1}{1-\rho^{-1}}\biggl(1-\rho^{-\ell'}-\frac{\ell'}{a_{\ell'}}\biggr)+1\\
 & \geq & \frac{1}{1-\rho^{-1}}(1-\varepsilon)+1\\
 & = & \frac{\rho}{\rho-1}\bigl(1-\bigl(1-\frac{\rho-1}{\rho}\varphi\bigr)\bigr)+1\\
 & = & \varphi+1,
\end{eqnarray*}
where the second inequality uses \eqref{eq:epsilon_estimate_1} and \eqref{eq:epsilon_estimate_2}. This yields a contradiction to the fact that $\rho<\varphi+1$. Thus,
there is no $\rho$-competitive algorithm for $\rho<\varphi+1$.
\end{proof}

This result immediately yields the desired lower bound.
\begin{corollary}
For monotone, $1$-bounded, and fractionally subadditive objectives,
the knapsack problem~(\ref{eq:generic-opt}) does not admit a $\rho$-competitive
incremental solution for $\rho<\varphi+1$.
\end{corollary}

It is possible to define a problem instance of the potential-based
maximum flow problem on parallel edges which reflects the construction
above. Thus, the lower bound on the competitive ratio translates also to this special case.

\begin{corollary}
	The incremental maximum potential-based flow problem on parallel edges does not admit a $\rho$-competitive algorithm for $\rho<\varphi+1$.
\end{corollary}
\begin{proof}
	For $i=1, \dots, n$ define $E_i$ to be a set of $i$ parallel edges
	from $s$ to $t$ with unit capacities. For $e_i \in E_i$, define its
	resistance to be $\beta(e_i):=\varepsilon^{i}$ for some
	$0<\varepsilon<1$. Let the potential loss function $\psi$ be continuous
	and strictly increasing with $\psi(0)=0$. Let $p_i:= \varepsilon^i
	\psi(1)$ be the potential difference between~$s$ and~$t$ inducing a flow
	of $1$ on all edges $e \in E_i$.
	Then, the maximum potential-based flow
	on a subset	$S \subseteq E = \bigcup_{i=1}^{n} E_i$ is given by
	\begin{align*}
		f'(S)
		& =  \max \Biggl\{\sum_{e\in T}\psi^{-1}\!\biggl(\frac{p}{\beta(e)}\biggr)\,\Bigg\vert\, T \subseteq S, p \in \mathbb{R}_{\geq 0} \:\text{with}\: \psi^{-1}\!\biggl(\frac{p}{\beta(e)}\biggr)\!\!\leq\! u(e) \:\text{for all}\: e\in T\Biggr\} \\
		& =  \max \Biggl\{ \sum_{j=1}^{i} \sum_{e\in E_j \cap S} \psi^{-1} \left(\frac{p_i}{\beta(e_j)} \right) \; \Bigg\vert\; i \in \{1,\dots,n\} \Biggr\} \\
		&= \max_{i \in \{1,\dots,n\}} |S \cap E_i|+ \sum_{j=1}^{i-1} \psi^{-1}(\varepsilon^{i-j} \psi(1)) |S \cap E_j|.
	\end{align*}
	The weights that represent the construction cost of the edges are
	defined as in the problem instance above to be $w(e_i)=b+i!$ for
	$b=(n+2)!$.

	Assume there is a $\rho$-competitive algorithm with $\rho<\varphi+1$
	for this problem. Let $\varepsilon'>0$ with
	$\rho+\varepsilon'<\varphi+1$. By \Cref{thm:lower_bound_phi+1} there is
	an $n\in \N$ such that the incremental knapsack problem given above
	does not admit a $(\rho+\varepsilon')$-competitive solution for the
	objective function $f$ defined in
	\eqref{eq:definition_of_lower_bound_objective_function}.
	Choose $\varepsilon$ small enough such that
	$\rho n^2\psi^{-1}(\varepsilon\psi(1))<\varepsilon'$. This implies
	\begin{align*}
		f'(S)-f(S) = \sum_{j=1}^{i-1} \psi^{-1}(\varepsilon^{i-j} \psi(1))
		|S \cap E_j| \leq n^{2}\psi^{-1}(\varepsilon\psi(1))
		< \frac{\varepsilon'}{\rho}.
	\end{align*}
	Let $\pi$ be a $\rho$-competitive solution for
	$f'$ produced by the algorithm. Then, for $C \geq b+1$ we have
	\begin{align*}
		(\rho+\varepsilon')f(\pi(C))
		&> \rho \left(f'(\pi(C))-\frac{\varepsilon'}{\rho}\right)+ \varepsilon'
		\geq (f')^{*}(\pi(C))
		\geq f^{*}(\pi(C)),
	\end{align*}
	where in the first inequality we use $f(\pi(C))\geq1$ and $f'(S)-f(S)<\varepsilon'/\rho$, in the second we use $\rho$-competitiveness of $\pi$ w.r.t. $f'$ and in the third inequality we use the fact that $f(S) \leq f'(S)$.
	Thus,~$\pi$ would be a $(\rho+\varepsilon')$-competitive
	solution for the incremental knapsack problem, contradicting our
	assumption. Therefore, a $\rho$-competitive algorithm for the incremental
	maximum potential-based flow problem cannot exist for
	$\rho<\varphi+1$.
\end{proof}

\section{Application to flows\label{sec:flows}}

In this section we show that our algorithm $\algscale$ can be used
to solve problems as given in Example~\ref{exa:supply}.
Formally, for the \emph{incremental maximal potential-based flow problem} on parallel edges, we 
are given a graph $G=(V,E)$ consisting of two nodes $s$ and $t$
with a collection of edges between them, and want to determine an
order in which to build the edges while maintaining a potential-based
flow between~$s$ and~$t$ that is as large as possible.
To this end, we are given a continuous and strictly increasing potential-loss function $\psi : \mathbb{R} \to \mathbb{R}$ with $\lim_{x \to \infty} \psi(x) = \infty$. Every edge~$e$ has an edge resistance $\beta(e) \in \mathbb{R}_{>0}$ and a capacity $u(e)$. Vertex potentials $p_s, p_t \in \mathbb{R}$ induce a flow of $x_e = \psi^{-1}(p/ \beta(e))$ on edge~$e$ where $p = p_t - p_s$. This flow is only feasible if $x_e \leq u(e)$.
The goal is to choose vertex potentials $p_s,p_t \in \mathbb{R}$ together with a subset of active edges that maximizes the total induced flow. This yields the objective
\begin{align}
\label{eq:potential-flow-objective}
f(S)=\max \Biggl\{\sum_{e\in T}\psi^{-1}\biggl(\frac{p}{\beta(e)}\biggr)\,\Bigg\vert\, T \subseteq S, p \in \mathbb{R}_{\geq 0} \:\text{with}\: \psi^{-1}\biggl(\frac{p}{\beta(e)}\biggr)\!\!\leq\! u(e) \:\text{for all}\: e\in T\Biggr\}
\end{align}
for all $S \subseteq E$.
The function $f$ is obviously monotone. We
further obtain that $f$ scaled by $(\min_{e \in E} u(e))^{-1}$ is
$M$-bounded for $M: = \frac{\max_{e \in E} u(e)}{\min_{e \in E} u(e)}$ because $f(\{e\})= u(e)$.
We proceed to show that the objective is fractionally
subadditive.
\begin{proposition}
\label{prop:pot-based_max-flow_is_fractionally-subadditive}The function $f : 2^E \to \mathbb{R}_{\geq 0}$ defined in \eqref{eq:potential-flow-objective} is fractionally subadditive.
\end{proposition}

\begin{proof}
For $e\in E$, let
\[
p_{e}:=\beta(e)\psi(u(e))
\]
be the maximum potential difference between $s$ and $t$ such that
the flow along $e$ induced by the potential difference $p_{e}$ is
still feasible, i.e., does not violate the capacity constraint $u(e)$.
For $e,e'\in E$, we define $x_{e}(p_{e'})$ to be the flow value
along $e$ induced by a potential difference of $p_{e'}$ between~$s$
and~$t$ if this flow is feasible and $0$ otherwise. For $S\subseteq E$,
we have
\begin{align*}
f(S) & =  \max \Biggl\{\sum_{e\in T}\psi^{-1}\biggl(\frac{p}{\beta(e)}\biggr)\,\Bigg\vert\, T \subseteq S, p \in \mathbb{R}_{\geq 0} \:\text{with}\: \psi^{-1}\biggl(\frac{p}{\beta(e)}\biggr)\!\!\leq\! u(e) \:\text{for all}\: e\in T\Biggr\} \\
 & =  \max \Biggl\{ \sum_{e\in S}x_{e}(p_{e'}) \; \Bigg\vert\; e'\in E \Biggr\},
\end{align*}
i.e., $f$ is an XOS-function and thus fractionally subadditive (see \Cref{exa:XOS}).
\end{proof}

As a corollary, we obtain the following result.

\begin{corollary}
The incremental maximal potential-based flow problem on parallel edges admits a $\rho$-competitive solution with
\begin{align*}
\rho \in \bigl[ \max \bigl\{\varphi\!+\!1,M \bigr\}, \max \bigl\{3.293\sqrt{M}, 2M \bigr\} \bigr],
\end{align*}
 where $M = \frac{\max_{e \in E} u(e)}{\min_{e \in E} u(e)}$.
\end{corollary}

We now return to the \emph{incremental maximum flow problem} discussed in \Cref{sec:introduction}.
In this problem,
we are given a directed graph $G=(V,E)$ with two designated vertices $s,t \in V$. For $v \in \mathbb{R}_{\geq 0}$, a vector $(x_e)_{e \in E}$ is an $s$-$t$-flow of value $v$ if $x_e \leq u(e)$ for all $e \in E$ and
\begin{align*}
\sum_{e \in \delta^+(v)} x_e - \sum_{e \in \delta^-(v)} x_e =
\begin{cases}
\phantom{-}v & \text{ if } v = s,\\
-v 	& \text{ if } v = t,\\
\phantom{-}0 & \text{ otherwise},
\end{cases}
\end{align*}
where
\begin{align*}
	\delta^+(v) &= \{e \in E \mid e = (v,w) \text{ with } w \in V\},\\
	\delta^-(v) &= \{e \in E \mid e = (w,v) \text{ with } w \in V\}
\end{align*}
denote the set of outgoing edges and the set of ingoing edges of a vertex~$v$, respectively. The incremental maximum flow problem has the objective
\[
f(S)=\max \bigl\{v \;\big\vert\;\textrm{there exists an }s\textrm{-}t\textrm{-flow of value }v\textrm{ in } G_S = (V,S)\} \quad \text{ for all } S \subseteq E.
\]
It is straightforward to verify that $f$ is modular (and, hence, also fractionally subadditive) for the case that $G$ has only the two vertices $s$ and $t$ and all edges go from $s$ to $t$. We here consider the case for a general graph $G$. For this case, it is easy to see that the objective need not to be fractionally subadditive in general. In fact, for the example of \Cref{subfig:inc-flow-1}, we have
\begin{align*}
f(\{b\}) &= 0, & f(\{c\}) &= 0, & f(\{b,c\}) &= k. 	
\end{align*}
This contradicts fractional subadditivity for the choices $A =\{b,c\}$, $B_1 = \{b\}$, $B_2 = \{c\}$,\linebreak and $\alpha_1 = \alpha_2 = 1$.

We proceed to show that despite the lack of (fractional) subadditivity, this problem has a competitive solution when  $u(e) = 1$ for all $e \in E$.
To solve this problem, we describe the algorithm $\textsc{Quickest-Increment}$ that has been introduced by Kalinowski et al.~\cite{KalinowskiMatsypuraSavelsbergh/15} for a different incremental flow problem where the sum of the flow values for all integer capacities $C$ is to be maximized.  The algorithm starts
by adding the shortest path and then iteratively adds the smallest
set of edges that increase the maximum flow value by at least 1. Let
$r\in\N$ be the number of iterations until $\textsc{Quickest-Increment}$
terminates. For $i\in\{0,1,\dots,r\}$, let $\lambda_{i}$ be the size
of the set added in iteration~$i$, i.e., $\lambda_{0}$ is the length
of the shortest $s$-$t$-path, $\lambda_{1}$ the size of the set
added in iteration 1, and so on. For $k\in\{1,\dots,|E|\}$, we denote
the solution of size $k$ of the algorithm by $\algsol(k)$.

With $v_{\textrm{max}}\in\R_{\geq0}$ defined as the maximal possible
$s$-$t$-flow value in the underlying graph, for $j\in\{1,\dots,\lfloor v_{\textrm{max}}\rfloor\}$,
we denote by $c_{j}$ the minimum number of edges required to achieve
a flow value of at least $j$.
The
values $\lambda_{i}$ and $c_{j}$ are related in the following way; see Kalinowski et al.~\cite{KalinowskiMatsypuraSavelsbergh/15} (Lemma 4).
\begin{lemma}
\label{lem:MaxFlow_unit-capacity_estimate}
When $w(e) = u(e) = 1$ for all $e \in E$, we have $\lambda_{i}\leq c_{j}/(j-i)$ for all $i,j\in\N$
with $0\leq i<j\leq r$.
\end{lemma}

Using this estimate, we can find a bound on the competitive ratio
of \textsc{Quick\-est-In\-cre\-ment} for the unit weight
and unit capacity case.
\begin{theorem}
\label{thm:normal_flow_unit_cap_unit_weight}
For the incremental maximal flow problem with $w(e) = u(e) = 1$ for all $e \in E$, the algorithm $\textsc{Quickest-Increment}$ is 2-competitive.
\end{theorem}

\begin{proof}
Note that, since we consider the unit capacity case, we have $v_{\textrm{max}}=r+1$
because $\textsc{Quickest-Increment}$ increases the value of the
solution by exactly 1 in each iteration.

Consider some size $k\in\{1,\dots,|E|\}$. If $k<c_{1}$, we have $f\bigl(\optsol(k)\bigr)=0$,
i.e., every solution is competitive. If $k\geq c_{1}$, let $j:=f\bigl(\optsol(k)\bigr)$.
Note that we have $f\bigl(\optsol(c_{j})\bigr)=j=f\bigl(\optsol(k)\bigr)$ and therefore $k\geq c_{j}$.
By \Cref{lem:MaxFlow_unit-capacity_estimate}, we have
\begin{eqnarray}
\sum_{i=0}^{\lceil j/2 \rceil-1}\lambda_{i} & \leq & \sum_{i=0}^{\lceil j/2 \rceil-1}\frac{c_{j}}{j-i}\nonumber \\
 & = & c_{j}\!\!\sum_{i=0}^{\lceil j/2 \rceil-1}\frac{1}{j-i}\nonumber \\
 & \leq & c_{j}\!\!\sum_{i=0}^{\lceil j/2\rceil-1}\frac{1}{j-\bigl\lceil\frac{j}{2}\bigr\rceil+1}\nonumber \\
 & = & c_{j}\biggl\lceil\frac{j}{2}\biggr\rceil\frac{1}{\bigl\lfloor\frac{j}{2}\bigr\rfloor+1}\leq c_{j}\label{eq:MaxFlow_Quickest-Increment_step-sizes}
\end{eqnarray}
This implies $f(\algsol(k))\geq f(\algsol(c_{j}))\overset{\eqref{eq:MaxFlow_Quickest-Increment_step-sizes}}{\geq}\bigl\lceil\frac{j}{2}\bigr\rceil\geq\frac{1}{2}j=\frac{1}{2}f(\optsol(k))$.
\end{proof}

Now, we turn to the case of unit capacities and rational weights.
By rescaling the weights, we can assume that, without loss of generality,
the weights are integral. To transform an instance with integral weights
to one where all edges have unit weight, one can simply replace every
edge $e\in E$ by a path of length $w(e)$ where every edge on the
new path has unit weight. Then, \Cref{thm:normal_flow_unit_cap_unit_weight}
can be applied and we obtain the following.
\begin{corollary}
For the incremental maximal flow problem with $u(e) = 1$ and $w(e) \in \mathbb{Q}_{\geq 0}$ for all $e \in E$, the algorithm $\textsc{Quickest-Increment}$ is $2$-competitive.
\end{corollary}

If we consider capacities that are in the interval $[1,M]$, one path
can bring at most~$M$ times as much flow as every other path. Combining this
with the fact that the solution of $\textsc{Quickest-Increment}$ for the instance
with $u(e)=1$ for all $e\in E$ is $2$-competitive yields that adding the edges in
the same order is always within a factor of $2M$ of the optimum solution.
\begin{corollary}
For the incremental maximal flow problem with $u(e) \in [1,M]$, $w(e) \in \mathbb{Q}_{\geq 0}$ for all $e \in E$, the solution generated by $\textsc{Quickest-Increment}$, if it only considers capacities $u(e)=1$, is $2M$-competitive.
\end{corollary}

As it turns out, the competitiveness of $\textsc{Quickest-Increment}$ of $2$ in the unit capacity case is optimal.
\begin{theorem}
\label{thm:MaxFlow_without_potentials_lower_bound}
For the incremental maximal flow problem with $u(e) = w(e) = 1$ for all $e \in E$, there is no $\rho$-competitive algorithm with $\rho < 2$.
\end{theorem}

\begin{proof}
Consider the graph $G=(V,E)$ with
\begin{eqnarray*}
V & := & \{s,t,u_{1},u_{2},u_{3},v_{1},v_{2},v_{3}\},\\
E & := & \{(s,u_{1}),(s,v_{1}),(u_{1},u_{2}),(v_{1},v_{2}),(u_{2},u_{3}),(v_{2},v_{3}),(u_{3},t),(v_{3},t),(u_{1},v_{3})\},
\end{eqnarray*}
with unit capacities and unit weights (cf.~\Cref{fig:MaxFlow_without_potentials_lower_bound}).
Let $\pi$ be an arbitrary incremental solution that is $\rho$-competitive. If the first three elements in $\pi$ are not the elements $(s,u_{1})$, $(u_1,v_3)$, and $(v_{3},t)$ (in any order) then the solution is not competitive for $C = 3$. Thus, any competitive solution contains these three elements at the first three positions. This, however, implies that the first eight elements of $\pi$ cannot contain the elements of the upper and lower paths, i.e., we have $\{(s,u_{1}),(u_1,u_{2}),(u_2,u_{3}),(u_3,t)\} \cup \{(s,v_{1}),(v_1,v_{2}),(v_2,v_{3}),(v_3,t)\}$. This implies that $f(\pi(8)) \leq 1$. Since $f^*(8) = 2$, we obtain $\rho \geq 2$, as claimed.
\end{proof}

\begin{figure}
\begin{center}
\begin{tikzpicture}
	\begin{scope}[every node/.style={circle,thick,draw}]
		\node[node,label=left:{$s$}] (s) at (0,1) {};
		\node[node,label=below:{$v_1$}] (v1) at (2,0) {};
		\node[node,label=below:{$v_2$}] (v2) at (4,0) {};
		\node[node,label=below:{$v_3$}] (v3) at (6,0) {};
		\node[node,label=above:{$u_1$}] (u1) at (2,2) {};
		\node[node,label=above:{$u_2$}] (u2) at (4,2) {};
		\node[node,label=above:{$u_3$}] (u3) at (6,2) {};
		\node[node,label=right:{$t$}] (t) at (8,1) {};
	\end{scope}
	
	\begin{scope}[edge/.style={->,thick}]
		\draw[edge] (s) to (u1);
		\draw[edge] (u1) to (u2);
		\draw[edge] (u2) to (u3);
		\draw[edge] (u3) to (t);
		\draw[edge] (s) to (v1);
		\draw[edge] (v1) to (v2);
		\draw[edge] (v2) to (v3);
		\draw[edge] (v3) to (t);
		\draw[edge] (u1) to (v3);
	\end{scope}
\end{tikzpicture}
\end{center}

\caption{\label{fig:MaxFlow_without_potentials_lower_bound}A lower bound instance
with best possible competitive ratio 2 for the problem $\textsc{Incremental}$
$\textsc{Maximum}$ $s$-$t$-$\textsc{Flow}$}
\end{figure}
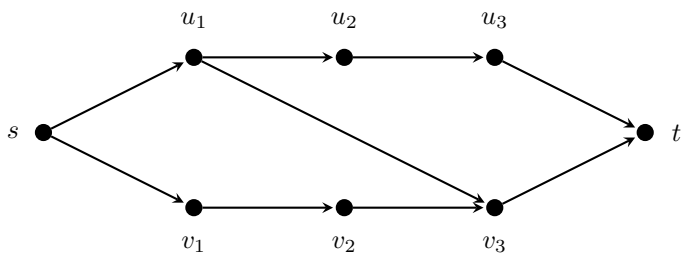

Furthermore, similar to the incremental maximization of a fractionally
subadditive function subject to a knapsack constraint, no algorithm
can have a better competitive ratio than~$M$ when $u(e) \in [1,M]$ for all $e \in E$.
\begin{theorem}
\label{thm:MaxFlow_without_potentials_lower_bound-1}
For the incremental maximal flow problem with $u(e) \in [1,M]$ and $w(e) =1$ for all $e \in E$, there is no $\rho$-competitive algorithm with $\rho < M$. 
\end{theorem}

\begin{proof}
Consider the graph $G=(V,E)$ with
\begin{eqnarray*}
V & := & \{s,t,v\},\\
E & := & \{(s,t),(s,v),(v,t)\},
\end{eqnarray*}
with unit weights and capacities $u((s,t))=1$, $u((s,v))=u((v,t)) = M$
(cf. Figure~\ref{subfig:inc-flow-1}).

Let $\pi$ be an arbitrary incremental solution. If $\pi$ does not begin with element $(s,t)$, then it is not competitive for $C=1$. This, however, implies that for any competitive incremental solution $\pi$, we have $\pi(2) \neq \{(s,v),(v,t)\}$. Thus, for any competitive $\pi$, we have $f(\pi(2)) = 1$ while $f^*(2) = M$. This implies $\rho \geq M$, as claimed.
\end{proof}

\bibliographystyle{plain}
\bibliography{../bibliography}

\begin{thebibliography}{10}

\bibitem{AmanatidisBirmpasMarkakis/16}
Georgios Amanatidis, Georgios Birmpas, and Evangelos Markakis.
\newblock Coverage, matching, and beyond: New results on budgeted mechanism
  design.
\newblock In {\em Proceedings of the 12th International Conference on Web and
  Internet Economics (WINE)}, pages 414--428, 2016.

\bibitem{AnariHNPST19}
Nima Anari, Nika Haghtalab, Seffi Naor, Sebastian Pokutta, Mohit Singh, and
  Alfredo Torrico.
\newblock Structured robust submodular maximization: Offline and online
  algorithms.
\newblock In {\em Proceedings of the 22nd International Conference on
  Artificial Intelligence and Statistics (AISTATS)}, pages 3128--3137, 2019.

\bibitem{BadanidiyuruDO12}
Ashwinkumar Badanidiyuru, Shahar Dobzinski, and Sigal Oren.
\newblock Optimization with demand oracles.
\newblock In {\em Proceedings of the 13th ACM Conference on Electronic Commerce
  (EC)}, pages 110--127, 2012.

\bibitem{BernsteinDisserGrossHimburg/20}
Aaron Bernstein, Yann Disser, Martin Gro\ss, and Sandra Himburg.
\newblock General bounds for incremental maximization.
\newblock {\em Math. Program.}, 191:953--979, 2020.

\bibitem{DisserKMS17}
Yann Disser, Max Klimm, Nicole Megow, and Sebastian Stiller.
\newblock Packing a knapsack of unknown capacity.
\newblock {\em {SIAM} J. Discrete Math.}, 31(3):1477--1497, 2017.

\bibitem{DisserWeckbecker/20}
Yann Disser and David Weckbecker.
\newblock Unified greedy approximability beyond submodular maximization.
\newblock arXiv:2011.00962, 2020.

\bibitem{DobzinskiSchapiry/06}
Shahar Dobzinski and Michael Schapira.
\newblock An improved approximation algorithm for combinatorial auctions with
  submodular bidders.
\newblock In {\em Proceedings of the 17th Annual ACM-SIAM Symposium on Discrete
  Algorithm (SODA)}, pages 1064--1073, 2006.

\bibitem{Feige98}
Uriel Feige.
\newblock A threshold of ln \emph{n} for approximating set cover.
\newblock {\em J. {ACM}}, 45(4):634--652, 1998.

\bibitem{Feige/09}
Uriel Feige.
\newblock On maximizing welfare when utility functions are subadditive.
\newblock {\em {SIAM} J. Comput.}, 39(1):122--142, 2009.

\bibitem{FujitaKM13}
Ryo Fujita, Yusuke Kobayashi, and Kazuhisa Makino.
\newblock Robust matchings and matroid intersections.
\newblock {\em {SIAM} J. Discrete Math.}, 27(3):1234--1256, 2013.

\bibitem{GrossPSSS/19}
Martin Gro{\ss}, Marc~E. Pfetsch, Lars Schewe, Martin Schmidt, and Martin
  Skutella.
\newblock Algorithmic results for potential-based flows: Easy and hard cases.
\newblock {\em Networks}, 73(3):306--324, 2019.

\bibitem{HartlineS07}
Jeff Hartline and Alexa Sharp.
\newblock Incremental flow.
\newblock {\em Networks}, 50(1):77--85, 2007.

\bibitem{HassinR02}
Refael Hassin and Shlomi Rubinstein.
\newblock Robust matchings.
\newblock {\em {SIAM} J. Discrete Math.}, 15(4):530--537, 2002.

\bibitem{KakimuraM13}
Naonori Kakimura and Kazuhisa Makino.
\newblock Robust independence systems.
\newblock {\em {SIAM} J. Discrete Math.}, 27(3):1257--1273, 2013.

\bibitem{Kakimura12}
Naonori Kakimura, Kazuhisa Makino, and Kento Seimi.
\newblock Computing knapsack solutions with cardinality robustness.
\newblock {\em Jpn. J. Ind. Appl. Math.}, 29(3):469--483, 2012.

\bibitem{KalinowskiMatsypuraSavelsbergh/15}
Thomas Kalinowski, Dmytro Matsypura, and Martin~W.P. Savelsbergh.
\newblock Incremental network design with maximum flows.
\newblock {\em Eur. J. Oper. Res.}, 242(1):51--62, 2015.

\bibitem{KawaseSF19}
Yasushi Kawase, Hanna Sumita, and Takuro Fukunaga.
\newblock Submodular maximization with uncertain knapsack capacity.
\newblock {\em {SIAM} J. Discrete Math.}, 33(3):1121--1145, 2019.

\bibitem{KlimmK22}
Max Klimm and Martin Knaack.
\newblock Maximizing a submodular function with bounded curvature under an
  unknown knapsack constraint.
\newblock In {\em Proceedings of the International Workshop on Approximation
  Algorithms for Combinatorial Optimization Problems (APPROX)}, pages
  49:1--49:19, 2022.

\bibitem{KobayashiT17}
Yusuke Kobayashi and Kenjiro Takazawa.
\newblock Randomized strategies for cardinality robustness in the knapsack
  problem.
\newblock {\em Theoret. Comput. Sci.}, 699:53--62, 2017.

\bibitem{LehmannLehmannNisan/06}
Benny Lehmann, Daniel Lehmann, and Noam Nisan.
\newblock Combinatorial auctions with decreasing marginal utilities.
\newblock {\em Games Econom. Behav.}, 55(2):270--296, 2006.

\bibitem{LinNRW10}
Guolong Lin, Chandrashekhar Nagarajan, Rajmohan Rajaraman, and David~P.
  Williamson.
\newblock A general approach for incremental approximation and hierarchical
  clustering.
\newblock {\em {SIAM} J. Comput.}, 39(8):3633--3669, 2010.

\bibitem{Marchetti-SpaccamelaV95}
Alberto Marchetti{-}Spaccamela and Carlo Vercellis.
\newblock Stochastic on-line knapsack problems.
\newblock {\em Math. Program.}, 68:73--104, 1995.

\bibitem{MatuschkeSS18}
Jannik Matuschke, Martin Skutella, and Jos{\'{e}}~A. Soto.
\newblock Robust randomized matchings.
\newblock {\em Math. Oper. Res.}, 43(2):675--692, 2018.

\bibitem{MegowM13}
Nicole Megow and Juli{\'{a}}n Mestre.
\newblock Instance-sensitive robustness guarantees for sequencing with unknown
  packing and covering constraints.
\newblock In {\em Proceedings of the 4th Conference on Innovations in
  Theoretical Computer Science (ITCS)}, pages 495--504, 2013.

\bibitem{NavarraP17}
Alfredo Navarra and Cristina~M. Pinotti.
\newblock Online knapsack of unknown capacity: How to optimize energy
  consumption in smartphones.
\newblock {\em Theoret. Comput. Sci.}, 697:98--109, 2017.

\bibitem{NemhauserWolseyFisher/78}
George~L. Nemhauser, Laurence~A. Wolsey, and Marshall~L. Fisher.
\newblock An analysis of approximations for maximizing submodular set functions
  - {I}.
\newblock {\em Math. Program.}, 14:265--294, 1978.

\bibitem{Nisan/00}
Noam Nisan.
\newblock Bidding and allocation in combinatorial auctions.
\newblock In {\em Proceedings of the 2nd {ACM} Conference on Electronic
  Commerce (EC)}, pages 1--12, 2000.

\bibitem{OrlinSU18}
James~B. Orlin, Andreas~S. Schulz, and Rajan Udwani.
\newblock Robust monotone submodular function maximization.
\newblock {\em Math. Program.}, 172(1-2):505--537, 2018.

\bibitem{Sviridenko/04}
Maxim Sviridenko.
\newblock A note on maximizing a submodular set function subject to a knapsack
  constraint.
\newblock {\em Oper. Res. Lett.}, 32:41--43, 2004.

\bibitem{ThielenTW16}
Clemens Thielen, Morten Tiedemann, and Stephan Westphal.
\newblock The online knapsack problem with incremental capacity.
\newblock {\em Math. Methods Oper. Res.}, 83(2):207--242, 2016.

\bibitem{Yoshida19}
Yuichi Yoshida.
\newblock Maximizing a monotone submodular function with a bounded curvature
  under a knapsack constraint.
\newblock {\em {SIAM} J. Discrete Math.}, 33(3):1452--1471, 2019.

\end{thebibliography}

\end{document}